\documentclass[lettersize,journal]{IEEEtran}
\usepackage{amsmath,amsfonts}
\usepackage{algorithmic}
\usepackage{algorithm}
\usepackage{array}
\usepackage[caption=false,font=normalsize,labelfont=sf,textfont=sf]{subfig}
\usepackage{textcomp}
\usepackage{stfloats}
\usepackage{url}
\usepackage{verbatim}
\usepackage{graphicx}
\usepackage[acronym]{glossaries}
\usepackage{amsmath}
\usepackage{amssymb}
\usepackage{bm}
\usepackage{amsthm}
\usepackage{mathtools}

\usepackage[table]{xcolor}
\usepackage{array}
\usepackage{multirow}
\usepackage{makecell}
\usepackage{amssymb}
\definecolor{graycell}{RGB}{217,217,217}
\definecolor{redcell}{RGB}{199,97,116}
\definecolor{goldcell}{RGB}{232,166,0}
\definecolor{bluecell}{RGB}{86,178,222}
\definecolor{greencell}{RGB}{7,160,116}
\definecolor{lightgreen}{RGB}{79,187,153}

\newtheorem{theorem}{Theorem}

\newtheorem{remark}{Remark}
\newtheorem{lemma}{Lemma}

\newacronym{awgn}{AWGN}{Additive White Gaussian Noise}
\newacronym{aoi}{AoI}{Age of Information}
\newacronym{arq}{ARQ}{automatic repeat request}
\newacronym{fbl}{FBL}{Finite Blocklength}
\newacronym{bs}{BS}{Base Station}
\newacronym{snr}{SNR}{Signal-to-Noise Ratio}
\newacronym{iid}{i.i.d.}{identically and independently distributed}
\newacronym{mdp}{MDP}{Markov Decision Process}
\newacronym{pmf}{pmf}{probability mass function}
\newacronym{pomdp}{POMDP}{Partially Observable Markov Decision Process}
\newacronym{bsc}{BSC}{Binary Symmetric Channel}
\newacronym{mac}{MAC}{Multiple Access Channel}
\newacronym{bler}{BLER}{Block Error Probability}
\newacronym{dpsc}{DPSC}{Dinkelbach per-stage cost}

\usepackage[
  backend=biber,
  style=ieee,
  doi=true,
  url=false
]{biblatex}

\AtBeginBibliography{\footnotesize}

\begin{document}

\title{Open-Loop Transmission with Discrete
	Blocklengths: Characterization of AoI-Optimal Schedules}

\author{\IEEEauthorblockN{Mojan Wegener \textit{Student Member, IEEE} and Eduard A. Jorswieck \textit{Fellow, IEEE}}
	% <-this % stops a space
	\thanks{The presented research is based on the joint research project “6G-MIMOLink”, which is funded
		by the Federal Ministry of Research, Technology and Space (BMFTR, grant number
		16KIS2441). The work of M. Wegener and E. Jorswieck is supported in part by the
		Federal Ministry of Research, Technology and Space of Germany in the program of "Souverän.
		Digital. Vernetzt." Joint project XG-RIC, project identification number: 16KIS2435. The authors
		are responsible for the content of this publication.}}% <-this % stops a space
%\thanks{Manuscript received April 19, 2021; revised August 16, 2021.}}

\markboth{}%
{M. Wegener \MakeLowercase{\textit{et al.}}: Open-Loop Transmission with Discrete
	Blocklengths: Characterization of AoI-Optimal Schedules}


\maketitle

\begin{abstract}
	In this paper, we consider a general status
	update system, consisting of a source
	at which a physical process is monitored and a
	receiver. For transmission of status updates,
	a discrete set of blocklength error probability tuples
	$\mathcal{C} = \{(n_1, \epsilon_1), (n_2, \epsilon_2),
		\ldots, (n_K, \epsilon_K)\}$
	is available. A typical example is a communication
	system with a fixed set of modulation and coding schemes.
	We assume that no feedback channel for
	acknowledgments exists and thus investigate the problem
	of constructing deterministic schedules for the status updating system.
	The average \gls{aoi} and
	average exponential \gls{aoi} are appropriate performance
	metrics to capture the freshness
	of information.
	Interestingly, a static blocklength
	is shown to be suboptimal in general.
	While for the average \gls{aoi}
	the gains that can be achieved by using a
	non-trivial periodic schedule are relatively small,
	they are significant for the average exponential \gls{aoi}.
	Motivated by this, we prove conditions for single-blocklength
	optimality when two blocklengths are available.
	To achieve this,
	we use a Dinkelbach-like approach to reformulate the problem as an infinite-state
	transient average-cost \gls{mdp} and apply careful
	bounding techniques
	to derive the optimality conditions. As a converse, we analyze
	stationary randomized policies to certify the parameter regions
	where a non-trivial deterministic schedule outperforms the best constant schedule.
	Extensions to the case with more than two transmission
	lengths are provided.
	Numerical results show that the derived conditions cover a large portion of the parameter
	space and suggest that the optimality conditions for a constant schedule with the shorter
	blocklength, while sufficient, may also be necessary.
	As an example, we apply the results to the problem of age-optimal transmission
	with a small LDPC codebook, which illustrates the practical applicability.
\end{abstract}

\begin{IEEEkeywords}
	Age of Information, Status Update Systems, Scheduling, Finite Blocklength.
\end{IEEEkeywords}

\section{Introduction}\label{sec:Introduction}

In current and future wireless communication systems, there is an increasing
interest in applications that require timely delivery of status information, such as
autonomous driving, remote control, drones, and real-time monitoring \cite{akyildiz6GFutureWireless2020}.
Typically, the freshness of information has been quantified by using the
latency of the updates or the time between updates.
However, even in simple systems like an M/M/1 queuing system,
it has been shown that there is an inherent trade-off between the latency
of the updates and the time between updates \cite{kaulRealtimeStatusHow2012}.
To solve this trade-off in an optimal way, the concept
of \gls{aoi} has been introduced in \cite{kaulRealtimeStatusHow2012},
which is defined as the time elapsed since the generation of the last successfully received update.
Since then, \gls{aoi} has been widely studied for different system models
such as parallel servers \cite{kamEffectMessageTransmission2016}, \cite{yatesStatusUpdatesNetworks2018},
multiple sources
\cite{yatesAgeInformationRealTime2019} and systems with packet management
\cite{costaAgeInformationStatus2016a}.
A similar trade-off between update rate and
latency is observed in systems where
the blocklengths used
to transmit updates can be optimized. Long blocklengths
reduce the error probability of the updates, but increase
the latencies and
the time between updates.
On the other hand, shorter blocklengths
have a higher error probability, but reduce the
latencies and the time
between updates.

There is a growing body of literature that studies
this trade-off and designs optimal transmission schemes for \gls{aoi}-based performance metrics.
These works can be broadly divided into two categories: In the first category,
there is no feedback available or the feedback about successful transmissions is only
used for retransmitting updates independent of the current age.
As a consequence, the resulting age process can be analyzed
by applying queueing theory or renewal theory, and the design problem reduces to finding the optimal
system parameters, such as the blocklength \cite{wangAgeInformationShortPacket2019},
\cite{yuAverageAgeInformation2020}, \cite{sungAgeInformationAnalysis2022}, \cite{caoInformationAgeDelayCorrelation2021},
\cite{yuanDataFreshnessOptimization2021}, \cite{caoMultiplexingDiversityAoIOriented2023},
\cite{pengJointOptimizationReliability2025}, \cite{xieEvaluationAgeInformation2019},
\cite{sacAgeOptimalChannelCoding2018}.
In the second category, feedback is available and is used
for history-dependent adaptation of the transmission process at each transmission.
As a result, \gls{mdp} theory is applied to optimize the system \cite{ceranAverageAgeInformation2019},
\cite{hanOptimalBlocklengthAllocation2019}, \cite{wangAdaptivePowerRate2021},
\cite{chenPreemptingMinimizeAge2026}, \cite{yaoBattleRateError2020}.
In the first category, the optimization variables, for example the blocklength,
are often assumed to be continuous. Together with the second-order normal approximation from
\cite{polyanskiyChannelCodingRate2010}, which relates the transmission time to the
decoding error probability, an optimal blocklength can often be found.
In reality, however, the blocklength is not a continuous variable.
Instead, we have a set of available blocklengths and corresponding
block error probabilities.

In the absence of feedback,
one might expect a fixed blocklength to be optimal.
However, as we show in this paper, this is not true in general.
To get as much insight as possible, we keep the model very general:
We assume that we have a set of $K$ blocklength error probability tuples
$\mathcal{C} = \{(n_1, \epsilon_1), (n_2, \epsilon_2), \ldots, (n_K, \epsilon_K)\}$,
where the blocklengths $n_i$ are ordered increasingly and the error probabilities $\epsilon_i$ are
ordered decreasingly.
Note that while this model is well suited for transmission, it can also be applied
to other systems.

For example, a computation offloading problem
can be modeled by the two blocklengths $n_1$ and $n_2$ and the
corresponding error probabilities $\epsilon_1$ and $\epsilon_2$
can be interpreted as the probability of an unsuccessful local or remote computation, respectively.
Already in the simplest case with two available blocklength error probability tuples
$n_1, \epsilon_1$ and $n_2, \epsilon_2$, we have identified
parameter regimes where it is not optimal to always use the same blocklength.
This can occur for a very short blocklength with a high error probability like $n_1 = 59, \epsilon_1 = 0.85$
and a much longer blocklength with a much lower error probability like $n_2 = 280, \epsilon_2 = 0.01$.
With these parameters, the average age can be reduced from 423 to 413 by using a
periodic schedule of
the form $n_1, n_1, n_1, n_2, n_1, \ldots$. While
these parameters are not very realistic and the gain is quite small, this
shows that in general it is not optimal to always use the same blocklength.
It also raises the question of how the system behaves for other age-related performance metrics.
One popular example is an exponential penalty function on the age \cite{yatesAgeInformationIntroduction2021},
\cite{tripathiWhittleIndexApproach2024}, \cite{sunUpdateWaitHow2017}.
The exponential penalty function is motivated by the fact that the state estimation
error of unstable systems grows exponentially with the age \cite{yatesAgeInformationIntroduction2021},
\cite{champatiPerformanceCharacterizationUsing2019}, \cite{klugelAoIPenaltyMinimizationNetworked2019}.
Motivated by this, we also analyze the exponential penalty function.
In this case, the relative gains from using different blocklengths are much higher.
For example, with the parameters $n_1 = 50, \epsilon_1 = 0.35, n_2 = 121, \epsilon_2 = 0.02$
and $\alpha = 0.02$,
the average exponential age can be reduced from 60 to 48, which is a
significant reduction of around 20~\%.
Our main contributions are as follows:
\begin{itemize}
	\item To the best of our knowledge, this is the first work to show that
	      a constant blocklength can be suboptimal in feedback-free
	      systems when a discrete set of blocklengths is available.
	\item We derive closed-form expressions for the average and average exponential \gls{aoi}
	      for periodic schedules with different blocklengths.
	\item Motivated by the observation that a constant blocklength is not optimal in general,
	      we derive sufficient conditions in closed form, which establish optimality of
	      constant schedules for both the average and the average exponential \gls{aoi}.
	\item Based on stationary randomized policies, we derive a method
	      for determining the parameter regions where a non-trivial
	      deterministic schedule
	      outperforms the best constant schedule.
	\item In addition to extensive numerical evaluations of the results, we show how the
	      derived optimality conditions can be used to find the optimal transmission scheme
	      for a realistic communication scenario with a
	      LDPC codebook.
\end{itemize}

The rest of the paper is organized as follows:
The system model and the performance metrics are introduced in Section \ref{sec:system_model}.
In Section \ref{sec:PeriodicPoliciesAoI},
we derive expressions
for the average and average exponential \gls{aoi} for
periodic schedules.
In Section \ref{sec:OptimalityConditions}, the optimality conditions
for constant schedules are derived.
The conditions are generalized to an arbitrary number
of blocklengths in Section~\ref{sec:MultiBlocklengthOptimalityConditions}.
The stationary randomized policies are analyzed in Section \ref{sec:stationaryRandomizedPolicies}
to determine the parameter regions where a non-trivial deterministic schedule outperforms
the best constant schedule.
In Section \ref{sec:numericalResults}, we present numerical visualizations of the optimality regions and
a case study for a realistic communication scenario.
Finally, the conclusions are drawn in Section \ref{sec:conclusion}.

\noindent\textit{Notation:}
Sets are denoted by calligraphic letters, and boldface
letters denote vectors or
sequences. The expectation operator and the
indicator function are denoted by
$\mathbb{E}[\cdot]$ and $\mathbf{1}\{\cdot\}$, respectively.
For any scalar
quantity $x$, we use
$\bar{x} \coloneqq 1-x$.

\section{System Model}\label{sec:system_model}

We consider a single source-receiver pair. The source monitors
a physical process and sends time-stamped status updates to the receiver. The source is modeled
as a generate-at-will source, i.e., it can generate a fresh status update at any time instant.
To transmit a status update, the source needs to choose a blocklength from
a discrete set of available blocklengths.
We denote the set of available blocklength block error probability pairs by
$\mathcal{C} = \{(n_1, \epsilon_1), (n_2, \epsilon_2), \ldots, (n_K, \epsilon_K)\}$,
where $K$ is the number of available blocklengths.
We assume that the set is ordered such that
$ 0 <n_1 < n_2 < \ldots < n_K$ and
$ 1 >\epsilon_1 > \epsilon_2 > \ldots > \epsilon_K \geq 0$.
We denote the block error probability associated with blocklength
$n_j$ by the function
$\epsilon(n_j)=\epsilon_j,\quad j\in\{1,\ldots,K\}$.
Accordingly, $\bar{\epsilon}(n_j)=1-\epsilon(n_j)=\bar{\epsilon}_j$.
The blocklength chosen for the $i$-th update is denoted by $a_i\in\{n_1,\ldots,n_K\}$, $i=0,1,\ldots$
We assume that the source has no feedback about the success or failure
of previous transmissions.
At time $t=0$, we assume that the receiver has an initial successfully received
update of blocklength $n_1$, generated at time $-n_1$.
Hence, the initial age is $\Delta(0)=n_1$.
As a performance metric, we consider a non-decreasing
function of the \gls{aoi} $\Delta(t)$ at the receiver side.
The \gls{aoi} at time $t$ is defined as
\begin{equation}
	\Delta(t) = t - u(t),
\end{equation}
where $u(t)$ is the generation time stamp of the most recently received update at the receiver.
We also assume perfect error detection at the receiver side. Therefore, the age $\Delta(t)$ behaves as follows:
For the points in time where a new status update $i$ is successfully received,
the age drops to $a_i$, which is the used blocklength of the update.
Between two successful receptions, the age increases linearly with slope one.
This results in a sawtooth-like behavior of the age function $\Delta(t)$.
In this work, the two performance metrics
we consider are the time-average
age $\Delta_{\mathrm{avg}}$ and the time-average exponential age $\Delta_{\exp,\mathrm{avg}}$,
which are defined as
\begin{equation}
	\Delta_{\mathrm{avg}} = \limsup_{T' \to \infty} \frac{1}{T'} \mathbb{E}
	\left [ \int_0^{T'} \Delta(t) dt \right ],
\end{equation}
and
\begin{equation}
	\Delta_{\exp,\mathrm{avg}} = \limsup_{T' \to \infty} \frac{1}{T'}
	\mathbb{E} \left [ \int_0^{T'} e^{\alpha \Delta(t)} dt \right ],
\end{equation}
where $\alpha$ is a positive parameter that controls the exponential growth of the penalty function.
Next, we introduce the policy classes of interest in this paper.
We denote by $\Pi$ the set of all policies, possibly randomized and history-dependent.
Note that in our model, we have no feedback on successful transmissions, and
therefore the history of the system is only given by the chosen blocklengths in the past.
We denote by $\Pi_D \subseteq \Pi$ the set of all deterministic policies.
Consequently, every deterministic policy induces a fixed action sequence,
and we refer to the elements of $\Pi_D$ as deterministic schedules.
We denote by $\Pi_{SD} \subseteq \Pi_D$ the set of all constant schedules where the same blocklength
is used for every transmission. Finally, we denote by
$\Pi_{SR}$ the set of stationary randomized
policies where for each transmission the blocklength is chosen randomly according to
a fixed distribution.
In this work, we focus on the
class of deterministic schedules $\Pi_D$,
which can be easily implemented.
With these definitions, we can formally state the
main optimization
problems of this paper:
\begin{align}
	 & \Delta_{\mathrm{avg}}^* =  \inf_{\pi \in \Pi_D} \Delta_{\mathrm{avg}}(\pi)\tag{P1}\label{eq:AoIOptimizationProblem1}          \\
	 & \Delta_{\exp,\mathrm{avg}}^* = \inf_{\pi \in \Pi_D} \Delta_{\exp,\mathrm{avg}}(\pi)\tag{P2}\label{eq:AoIOptimizationProblem2}
\end{align}

\section{Average \gls{aoi} and Exponential \gls{aoi} for Periodic
  Deterministic Schedules}\label{sec:PeriodicPoliciesAoI}

In this section, we derive closed-form expressions for the average and average exponential
\gls{aoi} of deterministic periodic schedules
$s_0,s_1,\ldots,s_{P-1},s_0,s_1,\ldots$
For each position $i\in\{0,\ldots,P-1\}$, we define the success and
error probabilities as
$p_i = \bar{\epsilon}(s_i)$ and $q_i=\epsilon(s_i)$.
All schedule-position indices in this section are interpreted modulo $P$.
With this notation, the average \gls{aoi} is given as follows:

\begin{theorem}\label{theo:AoICyclicClosedForm}
	Let $\pi \in \Pi_D$ be a deterministic periodic schedule of period $P$.
	Then the average \gls{aoi} of this schedule is given by
	\begin{align}
		\Delta_{\mathrm{avg}}
		=
		\frac{
			\sum_{i=0}^{P-1}
			p_{i-1}
			\left(
			s_{i-1}\mathbb{E}[Y_i]
			+
			\frac{\mathbb{E}[Y_i^2]}{2}
			\right)
		}{
			\sum_{i=0}^{P-1}
			p_{i-1}\mathbb{E}[Y_i]
		},
	\end{align}
	with
	\begin{align}
		 & \mathbb{E}[Y_i]
		=
		\frac{
			\displaystyle
			\sum_{k=0}^{P-1}
			\left(\prod_{r=0}^{k-1}q_{i+r}\right)
			s_{i+k}
		}{
			\displaystyle
			1-\prod_{r=0}^{P-1}q_{i+r}
		}, \label{eq:EY}       \\
		 & \mathbb{E}[Y_i^2] =
		\frac{
			\sum_{k=0}^{P-1}
			\left(\prod_{r=0}^{k-1} q_{i+r}\right)\bigl(s_{i+k}\bigr)^2
		}{
			1-\prod_{r=0}^{P-1} q_{i+r}
		}\label{eq:EYSquared}
		\\
		 & \quad+
		\frac{
			2\sum_{k=0}^{P-1}
			\left(\prod_{r=0}^{k} q_{i+r}\right)
			s_{i+k}\,
			\mathbb{E}\!\left[Y_{i+k+1}\right]
		}{
			1-\prod_{r=0}^{P-1} q_{i+r}
		}.\notag
	\end{align}
\end{theorem}

The proof is given in Appendix \ref{app:AoICyclicClosedForm}.
In the next theorem, we give a closed-form expression for the average exponential \gls{aoi} of periodic schedules.
Note that, for the exponential \gls{aoi}, a stability condition
is needed, which ensures that the average exponential \gls{aoi} is finite.

\begin{theorem}\label{theo:ExponentialAoICyclicClosedForm}
	Let $\pi \in \Pi_D$ be a deterministic periodic schedule of period $P$.
	If the condition
	\begin{align}\label{eq:ExpoAgeExistenceCondition}
		e^{\alpha\sum_{j=0}^{P-1} s_j}\prod_{j=0}^{P-1} q_j
		<1
	\end{align}
	is satisfied, then the average exponential \gls{aoi} of this schedule is given by
	\begin{align}
		\Delta_{\exp,\mathrm{avg}} = \frac{
			\sum_{i=0}^{P-1} p_{i-1}\, e^{\alpha s_{i-1}}\left(M_i-1\right)
		}{
			\alpha \sum_{i=0}^{P-1} p_{i-1}\,\mathbb{E}[Y_i]
		},
	\end{align}
	where $\mathbb{E}[Y_i]$ is defined as in Theorem
	\ref{theo:AoICyclicClosedForm},
	\begin{align}
		M_i
		=
		\frac{
			\displaystyle \sum_{k=0}^{P-1}
			\left(\prod_{r=0}^{k-1} \rho_{i+r}\right)\, \gamma_{i+k}
		}{
			\displaystyle 1-\prod_{r=0}^{P-1} \rho_{i+r}
		},
	\end{align}
	and $\rho_i$ and $\gamma_i$ are defined as
	\begin{align}
		\rho_i  = q_i e^{\alpha s_i}, \qquad
		\gamma_i = p_i e^{\alpha s_i}.
	\end{align}
\end{theorem}

The proof is given in Appendix \ref{app:ExponentialAoICyclicClosedForm}.
We used these expressions to compare non-trivial periodic
schedules with constant schedules, in which the same fixed blocklength is used
for every transmission. While for most parameter
choices, we did not find a non-trivial
periodic schedule that outperforms the best
constant schedule, such schedules
do exist.
One example was presented in the introduction.
To gain further insight into when constant schedules are optimal,
we derive in the next section a set of
sufficient conditions, which can be easily checked,
under which a constant schedule $\pi \in \Pi_{SD}$ is optimal.

\section{Optimality Conditions with Two Blocklengths}\label{sec:OptimalityConditions}

In this section, we present one of the main
results of this paper, which consists of simple conditions
that guarantee the optimality of a constant schedule.
For ease of exposition, we focus on the simplest case of two available blocklengths with
$n_1 < n_2$ and corresponding block error probabilities $\epsilon_1 > \epsilon_2$.
Later we extend the results to the case of more than two blocklengths.
Our approach is as follows:
We first recast the problem as a semi-Markov infinite-state average-cost \gls{mdp}, where at each transmission
instance, the sender only knows the probability distribution of the current
age based on the previously chosen blocklengths and the starting state with $\Delta(0) = n_1$.
We then use a Dinkelbach-like approach to transform the average-cost
semi-\gls{mdp} into a transient infinite-state deterministic average-cost \gls{mdp}.
In the resulting average-cost \gls{mdp}, each stage cost depends on a
parameter $\lambda$ introduced by the Dinkelbach
approach. This parameter can be bounded by bounding the optimal average \gls{aoi}.
Finally, we show that for certain parameters $n_1, n_2, \epsilon_1, \epsilon_2$,
and the parameter $\lambda$ bounded as in the previous step,
any non-constant schedule can be improved
by replacing an action that differs from the
corresponding constant schedule. Repeating
this replacement for every differing action
yields the constant schedule without increasing
the total cost. Therefore, the
constant schedule is optimal.

We start with the \gls{aoi} metric. The results for
the exponential \gls{aoi} metric are very similar,
and the proofs are based on the same ideas.
As a first step, problem (\ref{eq:AoIOptimizationProblem1})
is reformulated as follows:
\begin{align}
	\Delta_{\mathrm{avg}}^* & = \inf_{\pi\in \Pi_D} \limsup_{T' \to \infty} \mathbb{E}\left[ \frac{1}{T'} \int_0^{T'} \Delta(t) dt \right].\label{eq:gObjective}     \\
	                        & = \inf_{\pi\in \Pi_D} \limsup_{T \to \infty} \frac{\mathbb{E}\left[\sum_{i=0}^{T}c(X_i,a_i)\right]}{\sum_{i=0}^{T}a_i}                 \\
	                        & = \inf_{\pi \in \Pi_D} \limsup_{T\to\infty} \frac{\sum_{i=0}^{T}\tilde{c}(b_i,a_i)}{\sum_{i=0}^{T}a_i}, \label{eq:DeterministicAoIMDP}
\end{align}
where the random variable $X_i$ denotes the age at the beginning of the $i$-th transmission
and $c(x_i,a_i)$ is the area under the age curve during the $i$-th transmission,
which is given by $c(x_i,a_i) = x_i a_i + \frac{a_i^2}{2}$.
The expected cost $\tilde{c}(b_i,a_i)$ is given by
\begin{align}
	\tilde{c}(b_i,a_i) = \sum_{x\in\mathcal{S}} b_i(x)\,c(x,a_i),
\end{align}
where $b_i$ is the belief on $X_i$ at the beginning of the $i$-th transmission and is given by
\begin{equation}
	b_i(x)\triangleq \mathbb{P}(X_i=x\mid a_0,\dots,a_{i-1}),\qquad x\in\mathcal{S}.
\end{equation}
Since the model is blind (no observations), the belief update is purely a
prediction step induced by the transition kernel. For an action $a_i\in \{n_1,n_2\}$,
the next belief is given by
\begin{equation}
	b_{i+1}(x')=\sum_{x\in\mathcal{S}} b_i(x)\,P(x'\mid x,a_i),\qquad x'\in\mathcal{S},
\end{equation}
with the following transition kernel:
\begin{align}
	P(x'\mid x,a_i) =
	\begin{cases}
		\bar{\epsilon}(a_i), & \text{if } x' = a_i.     \\
		\epsilon(a_i),       & \text{if } x' = x + a_i. \\
		0,                   & \text{otherwise}.
	\end{cases}
\end{align}
We would like to point out that the semi-\gls{mdp} in (\ref{eq:DeterministicAoIMDP})
is deterministic, has an infinite state space and does not have any recurrence
structure, because each state is visited exactly once. Therefore,
we cannot use the standard approach of Bellman optimality equations and value iteration
to solve the problem. Instead, we transform the problem into a standard average-cost
\gls{mdp} and then bound the difference in the total cost
when one action is changed.
By using a Dinkelbach approach similar to \cite{yaoBattleRateError2020}, %\cite{liSamplingAchieveGoal2024a}
the semi-\gls{mdp} can be transformed into a standard
average-cost \gls{mdp}:
\begin{lemma}\label{lem:DinkelbachTrafo}
	The semi-\gls{mdp} defined in (\ref{eq:DeterministicAoIMDP})
	is equivalent to the following optimization problem:
	\begin{align}
		\Delta_{\mathrm{avg}}^* \!=\! \inf \left\{\! \lambda\!: \!\!\inf_{\pi\in \Pi_D}\! \limsup_{T\to\infty} \frac{1}{T} \sum_{i=0}^{T} (\tilde{c}(b_i,a_i)
		- \lambda a_i) \leq 0\right\},\label{eq:TransformedDinkelbachProblem}
	\end{align}
	and for the $\lambda^*$ that solves the problem above, we have $\Delta_{\mathrm{avg}}^* = \lambda^*$.
\end{lemma}
The proof is given in Appendix \ref{app:DinkelbachTrafo}.
From now on we focus on the transformed problem in Lemma \ref{lem:DinkelbachTrafo},
which is visualized in Fig. \ref{fig:TransientMDPChain}.
\begin{figure}
	\centering
	\includegraphics[width=1.0\linewidth]{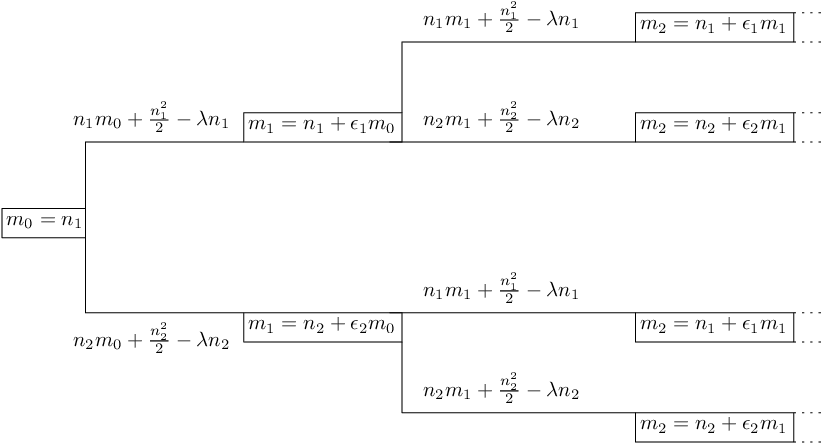}
	\caption{Transient deterministic chain of the belief MDP.}
	\label{fig:TransientMDPChain}
\end{figure}
The state $m_i$ of the \gls{mdp} is given by the expected age
at the beginning of the $i$-th transmission:
\begin{align}
	m_i \triangleq \sum_{x\in\mathcal{S}} b_i(x)x.
\end{align}
We denote the cost terms $a_i m_i +\frac{a_i^2}{2} - \lambda a_i$
as \gls{dpsc}.
Next, we investigate how the evolution
of the expected age $m_{i+1}$ depends on the current
expected age $m_i$ and the chosen action $a_i$.
In Lemma~\ref{lem:boundedcostLemmaAoI} we will show
that the progression of the expected age can be written as
$m_{i+1} = f_e(m_i), e \in \{1,2\}$, where
we have introduced the two maps $f_1(m)$ and $f_2(m)$ as
\begin{align}\label{eq:fAoIMaps}
	f_1(m)  = n_1 + \epsilon_1 m, \quad
	f_2(m)  = n_2 + \epsilon_2 m.
\end{align}
Based on $m_i$, the function $f_1$ gives the expected age $m_{i+1}$
when choosing the smaller blocklength $n_1$, while the
function $f_2$ gives the
expected age $m_{i+1}$ when choosing the larger blocklength $n_2$.
Details on this can
be found in Appendix \ref{app:boundedcostLemmaAoI}.
In the next lemma, we give asymptotic
bounds on the expected age $m_i$:

\begin{lemma}\label{lem:boundedcostLemmaAoI}
	For each $\delta > 0$, there exists a positive integer $i_0$
	large enough such that for all $i \geq i_0$, the expected age
	$m_i$ is bounded as
	\begin{align}
		m_L - \delta \leq \underbrace{\sum_{x\in \mathcal{S}} b_i(x)x}_{m_i} \leq m_U,
	\end{align}
	with
	\begin{align}
		m_L = \min \left\{ \frac{n_1}{\bar{\epsilon}_1}, \frac{n_2}{\bar{\epsilon}_2} \right\}, \quad
		m_U = \max \left\{ \frac{n_1}{\bar{\epsilon}_1}, \frac{n_2}{\bar{\epsilon}_2} \right\}.
	\end{align}
\end{lemma}
The proof is based on the fixed points
of $f_1$ and $f_2$ and can be found in Appendix \ref{app:boundedcostLemmaAoI}.
Note that from Lemma \ref{lem:boundedcostLemmaAoI}, it immediately follows that
$\tilde{c}(b_i,a_i)$ is uniformly bounded from below and above by positive constants.
Also note that the functions $f_1(m)$ and $f_2(m)$ cross at most once at
\begin{align}
	m = \frac{n_2-n_1}{\epsilon_1-\epsilon_2} \triangleq m^*.
\end{align}
Furthermore, the function $f_2(m)-f_1(m)= (n_2-n_1) + (\epsilon_2-\epsilon_1)m$ is strictly decreasing in $m$, and it follows that
\begin{align}
	 & f_1(m) < f_2(m) \text{ for } m < m^*, \\
	 & f_1(m) > f_2(m) \text{ for } m > m^*.
\end{align}
Depending on the current expected age $m_{i}$, we have now characterized
which action is better in terms of the expected age at the beginning of the next transmission.
Because each \gls{dpsc} is strictly increasing in $m_{i}$,
there is a trade-off between choosing actions that minimize the total cost in the long run
due to the progression of the expected ages $m_{i+1}, m_{i+2}, \ldots$ and actions
that minimize the
immediate \gls{dpsc} of the action itself.
To further analyze this trade-off, we have to bound
the parameter $\lambda$ which occurs in the \gls{dpsc}. According
to Lemma \ref{lem:DinkelbachTrafo},
$\lambda$ can be bounded by bounding the average age $\Delta_{\mathrm{avg}}^*$.
This is done in the following lemma:
\begin{lemma}\label{lem:AoIBounds}
	Let $\Delta_{\mathrm{avg}}^*$ be the minimal average \gls{aoi} achievable by any
	schedule $\pi \in \Pi_D$.
	Then the following bounds hold:
	\begin{align}
		 & n_1+\max \Bigg\{\frac{1}{2}n_1 \frac{1+\epsilon_2}{\bar{\epsilon}_2} , \\
		 & \min_{\substack{\theta_1,\theta_2\ge 0                                 \\
				\frac{n_1}{\bar{\epsilon}_1}\theta_1+
				\frac{n_2}{\bar{\epsilon}_2}\theta_2\le 1}}
		\Bigg\{
		n_1(n_2-n_1)\theta_2
		+
		\frac{1}{2(\theta_1+\theta_2)}
		\Bigg\} \Bigg\}                                                           \\
		 & \leq \Delta_{\mathrm{avg}}^*
		\leq \min \left\{ \frac{1}{2}n_1 \frac{1+\epsilon_1}{\bar{\epsilon}_1}+n_1,
		\frac{1}{2}n_2 \frac{1+\epsilon_2}{\bar{\epsilon}_2} + n_2 \right\}. \label{eq:lambdaUpperBound}
	\end{align}
\end{lemma}

The proof is given in Appendix \ref{app:AoIBounds}.
We now state the first optimality condition:
\begin{theorem}\label{thm:chooseN1OptimalUnified}
	Let
	\begin{align}
		m^* \!\!\triangleq \frac{n_2-n_1}{\epsilon_1-\epsilon_2},
		\,\,\,\delta_n \triangleq n_2\!-n_1,
		\,\,\,\delta_{\epsilon} \triangleq \epsilon_1\!-\epsilon_2,
	\end{align}
	and define the bound
	\begin{align}
		U & \!\triangleq \min \left\{
		\frac{1}{2}n_1\frac{1+\epsilon_1}{\bar{\epsilon}_1}\!+n_1,\,
		\frac{1}{2}n_2\frac{1+\epsilon_2}{\bar{\epsilon}_2}+n_2
		\right\}.
	\end{align}
	Assume one of the following two regimes holds:

	\medskip
	\noindent\textbf{(i) Above-threshold regime:}
	\begin{align}
		 & m_L > m^* \qquad \text{and} \notag                                                                                           \\
		 & U<\frac{n_1+n_2}{2}+\frac{n_1}{\bar{\epsilon}_1} + m_U\left( 1-\frac{n_1\delta_{\epsilon}}{\bar{\epsilon}_1\delta_n}\right).
	\end{align}

	\medskip
	\noindent\textbf{(ii) Below-threshold regime:}
	\begin{align}
		 & m_U < m^* \qquad \text{and} \notag                                                                                           \\
		 & U<\frac{n_1+n_2}{2}+\frac{n_1}{\bar{\epsilon}_1} + m_L\left( 1-\frac{n_1\delta_{\epsilon}}{\bar{\epsilon}_1\delta_n}\right).
	\end{align}
	Then the constant schedule that always chooses \(n_1\) is optimal among all deterministic schedules
	for the time-average age criterion.
\end{theorem}
The proof can be found in Appendix
\ref{app:chooseN1OptimalUnified}.
So far we have derived conditions under which always choosing $n_1$ is optimal. By using the same approach and reversing the lower and upper bounds,
we can obtain corresponding conditions under
which always choosing $n_2$ is optimal:

\begin{theorem}\label{thm:chooseN2OptimalUnified}
	Let $m^*$, $\delta_n$ and $\delta_{\epsilon}$
	be defined as in
	Theorem \ref{thm:chooseN1OptimalUnified}
	and define the bound
	\begin{align}
		L & \triangleq n_1+\max \Bigg\{\frac{1}{2}n_1 \frac{1+\epsilon_2}{\bar{\epsilon}_2} , \\
		  & \min_{\substack{\theta_1,\theta_2\ge 0                                            \\
				\frac{n_1}{\bar{\epsilon}_1}\theta_1+
				\frac{n_2}{\bar{\epsilon}_2}\theta_2\le 1}}
		\Bigg\{
		n_1(n_2-n_1)\theta_2
		+
		\frac{1}{2(\theta_1+\theta_2)}
		\Bigg\}. \Bigg\}
	\end{align}
	Assume that $m_L > m^*$ and
	\begin{align}
		L> \frac{n_1+n_2}{2} + \frac{n_2}{\bar{\epsilon}_2} + m_L \left( 1-\frac{n_2\delta_{\epsilon}}{\bar{\epsilon}_2\delta_n}\right).
	\end{align}
	Then the constant schedule that always chooses \(n_2\) is optimal
	among all deterministic schedules
	for the time-average age criterion.
\end{theorem}
The proof is similar to the proof of Theorem \ref{thm:chooseN1OptimalUnified} and can be found in Appendix
\ref{app:chooseN2OptimalUnified}.

\begin{remark}[No below-threshold regime for the constant \(n_2\) schedule]
	Assume the below-threshold condition $m_U<m^*$ holds.
	Then the constant schedule that always chooses \(n_2\) cannot be optimal.
	This can be seen by noting that the
	average age of a constant schedule with blocklength $n_e$ is given
	by
	\begin{align}
		J_e
		=
		n_e+\frac{1}{2}n_e\frac{1+\epsilon_e}{\bar{\epsilon}_e}
		=
		\frac{n_e}{\bar{\epsilon}_e}+\frac{n_e}{2}.
	\end{align}
	Some simple calculations show that $m_U < m^*$
	implies $J_2 - J_1 > 0$. This shows that
	the constant schedule with $n_2$ cannot be optimal in
	this regime.
\end{remark}

Numerical evaluations of Theorems
\ref{theo:AoICyclicClosedForm} and
\ref{theo:ExponentialAoICyclicClosedForm} indicate
that although constant schedules are generally not optimal
for the average \gls{aoi} metric, the gains that
can be achieved using non-constant schedules
are relatively small.
In comparison,
for the exponential age metric,
the gains that can be achieved using non-constant schedules
can be significant for certain parameter regimes.
We therefore
extend the previous results for the exponential
\gls{aoi} criterion
to better understand for which parameter regions
constant schedules are optimal.
The steps are similar to the average age case,
but with a different cost term in the corresponding
\gls{mdp} problem and with different equations
governing the evolution of the
expected exponential age.
The area under the curve of the exponential
age penalty during one transmission is given by
\begin{align}
	c_{\exp}(x_i,a_i) = \int_{0}^{a_i} e^{\alpha (x_i+t)} dt = \frac{1}{\alpha}e^{\alpha x_i}(e^{\alpha a_i}-1),
\end{align}
where $x_i$ denotes the age at the beginning of the transmission and
$a_i \in \{n_1, n_2\}$ is the chosen blocklength.
Hence, the cost $\tilde{c}_{\exp}(b_i,a_i)$ in our
infinite-state average-cost semi-\gls{mdp} is given by
\begin{align}
	\tilde{c}_{\exp}(b_i,a_i) = \frac{e^{\alpha a_i}-1}{\alpha} \sum_{x \in \mathcal{S}} b_i(x) e^{\alpha x}.
\end{align}

Following the same lines as in the proofs
of Theorems \ref{thm:chooseN1OptimalUnified} and \ref{thm:chooseN2OptimalUnified},
we obtain the following sufficient conditions for the optimality
of constant schedules for the exponential age criterion:
\begin{theorem}\label{thm:chooseN1OptimalExpUnified}
	Let
	\begin{align}
		 & K_e \triangleq \frac{e^{\alpha n_e}-1}{\alpha},\!\! \quad e\in\{1,2\}, \\
		 & \beta_e \triangleq \bar{\epsilon}_ee^{\alpha n_e},\!\! \quad
		\rho_e \triangleq \epsilon_e e^{\alpha n_e}, \!\!\quad e\in\{1,2\},
	\end{align}
	and define the following quantities:
	\begin{align}
		z^* \triangleq \frac{\beta_2-\beta_1}{\rho_1-\rho_2},
		\qquad
		\delta_n \triangleq n_2-n_1,
	\end{align}

	\begin{align}
		U_{\exp}
		\triangleq \min \left\{ \frac{K_1 \beta_1}{n_1 \bar{\rho}_1},
		\frac{K_2 \beta_2}{n_2 \bar{\rho}_2}\right\}
	\end{align}

	\begin{align}
		z_L = \min \left\{ \frac{\beta_1}{\bar{\rho}_1}
		, \frac{\beta_2}{\bar{\rho}_2}\right\}, \quad
		z_U =\max \left\{ \frac{\beta_1}{\bar{\rho}_1}
		, \frac{\beta_2}{\bar{\rho}_2}\right\}
	\end{align}
	If $\rho_1 < \rho_2 < 1$ or if
	$\rho_2 < \rho_1 < 1$ and one of the following two regimes holds:

	\medskip
	\noindent\textbf{(i) Above-threshold regime:}
	\begin{align}
		  & z_L > z^* \quad\text{and} \notag                                              \\
		U_{\exp}
		< &
		\frac{1}{\delta_n} \min_{z \in [z_L, z_U]} \Bigg\{ (K_2-K_1)z                     \\
		  & - \frac{K_1}{\bar{\rho}_1}\Big(\beta_1-\beta_2+(\rho_1-\rho_2)z\Big) \Bigg\}.
	\end{align}

	\medskip
	\noindent\textbf{(ii) Below-threshold regime:}
	\begin{align}
		  & z_U < z^* \quad\text{and}        \notag                                       \\
		U_{\exp}
		< &
		\frac{1}{\delta_n} \min_{z \in [z_L, z_U]} \Bigg\{ (K_2-K_1)z                     \\
		  & - \frac{K_1}{\bar{\rho}_1}\Big(\beta_1-\beta_2+(\rho_1-\rho_2)z\Big) \Bigg\}.
	\end{align}
	Then the constant schedule that always chooses $n_1$ is optimal
	among all deterministic schedules
	for the time-average exponential age criterion.
\end{theorem}
The proof can be found in Appendix \ref{app:chooseN1OptimalExpUnified}.

\begin{theorem}\label{thm:chooseN2OptimalExpUnified}
	Let $K_e$, $\beta_e$, $\rho_e$, $z^*$, $\delta_n$,
	$z_L$ and $z_U$ be defined as in
	Theorem \ref{thm:chooseN1OptimalExpUnified}, and define
	the bound
	\begin{align}
		 & L_{\exp}
		\triangleq
		\max\Bigg\{
		\frac{\bar{\epsilon}_2e^{\alpha n_1}
			\left(e^{\alpha n_1}-1\right)}
		{\alpha n_1\left(1-\epsilon_2e^{\alpha n_1}\right)}, \nonumber \\
		 & \min_{\substack{\theta_1,\theta_2\ge 0                      \\
				\frac{n_1}{\bar{\epsilon}_1}\theta_1+
				\frac{n_2}{\bar{\epsilon}_2}\theta_2\le 1}}
		\Bigg\{
		\frac{e^{\alpha n_1}}{\alpha}
		(\theta_1+\theta_2)
		\left(
		e^{\frac{\alpha}{\theta_1+\theta_2}}-1
		\right) \nonumber                                              \\
		 &
		+
		\frac{e^{\alpha n_2}-e^{\alpha n_1}}{\alpha}
		\left(e^{\alpha n_1}-1\right)\theta_2
		\Bigg\}
		\Bigg\}.
	\end{align}
	\begin{flalign*}
		         & \text{If} \quad \rho_2<\rho_1 < 1,
		\qquad
		z_L>z^*, &                                    &
	\end{flalign*}
	and the following condition holds:
	\begin{align}
		L_{\exp}
		>
		 & \frac{1}{\delta_n} \max_{z \in [z_L, z_U]} \Bigg\{\frac{K_2(\beta_2-\beta_1)}{\bar{\rho}_2} \\
		 & + z\left( (K_2-K_1) - \frac{K_2(\rho_1-\rho_2)}{\bar{\rho}_2}\right) \Bigg\}.
	\end{align}
	Then the constant schedule that always chooses $n_2$ is optimal
	among all deterministic schedules
	for the
	time-average exponential age criterion.
\end{theorem}
The proof can be found in Appendix \ref{app:chooseN2OptimalExpUnified}.
Note that there is no regime with $\rho_1 < \rho_2 < 1$
where always choosing $n_2$ is optimal. This
can be seen by comparing the two expressions
in the upper bound $U_{\exp}$. It follows that always
choosing $n_1$ achieves a smaller average
exponential \gls{aoi} than
always choosing $n_2$ if $\rho_1 < \rho_2 < 1$, and therefore
always choosing $n_2$ cannot be optimal among the
deterministic schedules in this case.

\section{Optimality Conditions with $K \geq 2$ Available Blocklengths}\label{sec:MultiBlocklengthOptimalityConditions}

In this section, we generalize the optimality
conditions from Section \ref{sec:OptimalityConditions}
to the case of $K$ blocklengths.
This requires only minor modifications and
results in the following conditions for
optimality:

\begin{theorem}
	Let $n_1<n_2<\ldots<n_K$ be the available blocklengths with the corresponding
	error probabilities $\epsilon_1>\epsilon_2>\ldots>\epsilon_K$.
	Let
	\begin{align}
		m_L & = \min_{i\in\{1,\ldots,K\}} \frac{n_i}{\bar{\epsilon}_i}, \label{eq:m_L_K2} \\
		m_U & = \max_{i\in\{1,\ldots,K\}} \frac{n_i}{\bar{\epsilon}_i}, \label{eq_m_U_K2} \\
		L   & = n_1+
		\max \Bigg\{
		\frac{1}{2}n_1\frac{1+\epsilon_K}{\bar{\epsilon}_K},\label{eq:L_K2}
		\\
		    &
		\min_{\substack{\theta_i\ge 0,\ i=1,\dots,K                                       \\
				\sum_{i=1}^K \frac{n_i}{\bar{\epsilon}_i}\theta_i\le 1}}
		\Bigg[
			n_1\sum_{i=1}^K (n_i-n_1)\theta_i
			+
			\frac{1}{2\sum_{i=1}^K \theta_i}
			\Bigg]
		\Bigg\}.\notag                                                                    \\
		U   & = \min_{i\in\{1,\ldots,K\}}
		\left\{ \frac{n_i}{2}\frac{1+\epsilon_i}{\bar{\epsilon}_i} + n_i\right\} \label{eq:U_K2}
	\end{align}
	Let the blocklength $n^* \in \{n_1, n_2, \ldots, n_K\}$ satisfy the following conditions:
	\begin{itemize}
		\item For all $n_i < n^*$ the conditions of Theorem \ref{thm:chooseN2OptimalUnified} are satisfied pairwise with
		      $n_1 = n_i$ and $n_2 = n^*$, where
		      the bounds $m_L$ and $m_U$ given in (\ref{eq:m_L_K2}), (\ref{eq_m_U_K2})
		      and the bounds $L$ and
		      $U$ given in (\ref{eq:L_K2}), (\ref{eq:U_K2})
		      are used.
		\item For all $n_i > n^*$ the conditions of Theorem \ref{thm:chooseN1OptimalUnified} are satisfied pairwise with
		      $n_1 = n^*$ and $n_2 = n_i$, where
		      the bounds $m_L$ and $m_U$ given in (\ref{eq:m_L_K2}), (\ref{eq_m_U_K2})
		      and the bounds $L$ and
		      $U$ given in (\ref{eq:L_K2}), (\ref{eq:U_K2})
		      are used.
	\end{itemize}
	Then the constant schedule with blocklength $n^*$ is
	optimal for the time-average age criterion
	within the class of deterministic schedules.

\end{theorem}

\begin{proof}
	The proof is a straightforward generalization of the proofs of Theorems \ref{thm:chooseN1OptimalUnified}
	and \ref{thm:chooseN2OptimalUnified}. If the conditions are satisfied,
	starting from an arbitrary schedule with
	blocklengths from the set
	$\{n_1, n_2, \ldots, n_K\}$,
	we can iteratively construct a sequence
	$\mu^{(0)}, \mu^{(1)}, \ldots, \mu^{(N_T)}$ of schedules
	to transform the arbitrary schedule into the schedule
	in which blocklength $n^*$ is used at every transmission.
	Because the conditions are satisfied
	pairwise, each replacement still strictly reduces
	the \gls{aoi} cost,
	and hence the final schedule
	with blocklength $n^*$ is optimal.
	The modified bounds on $L$ and
	$U$ are straightforward generalizations
	of the bounds in Theorems
	\ref{thm:chooseN1OptimalUnified}
	and \ref{thm:chooseN2OptimalUnified}.
\end{proof}
Note that a similar generalization can be done for the
average exponential age criterion.

\section{Conditions for Suboptimality of Constant Schedules}\label{sec:stationaryRandomizedPolicies}

In the previous sections,
we have derived conditions for
the optimality of constant schedules
within the class of deterministic schedules.
We have also seen that constant
schedules are not always optimal.
It is therefore natural to ask for conditions
under which a non-constant schedule is required for optimality
as a converse to Theorems \ref{thm:chooseN1OptimalUnified}, \ref{thm:chooseN2OptimalUnified},
\ref{thm:chooseN1OptimalExpUnified} and \ref{thm:chooseN2OptimalExpUnified}.
One straightforward approach is to calculate the average
age with the closed-form expressions from Section
\ref{sec:PeriodicPoliciesAoI}
and to compare it with the average age of
the constant schedules. However,
this does not result in an easily checkable
condition for
suboptimality of constant schedules.
In this section, we take
a different approach and analyze the average
and average exponential age
for stationary randomized policies. This results in
algebraic conditions that are sufficient
for the suboptimality of constant schedules
within the class of deterministic schedules.

\subsection{Stationary Randomized Policies for \gls{aoi}}
First, we derive a closed-form expression for the average \gls{aoi} for
stationary randomized policies. We then show that the optimal stationary randomized policy can be found by solving a simple
quadratic equation.

\begin{lemma}\label{lem:OptimalStationaryAoI}
	Let $\pi \in \Pi_{SR}$ be a stationary randomized policy that chooses $n_1$ with probability $\xi$ and
	$n_2$ with probability $\bar{\xi}$.
	Then the average \gls{aoi} $\Delta_{\mathrm{avg}}$ of $\pi$ is given by
	\begin{align}
		\Delta_{\mathrm{avg}} \!=\! \frac{\xi n_1 + \bar{\xi}n_2}{\xi\bar{\epsilon}_1+\bar{\xi}\bar{\epsilon}_2}
		\!+\! \frac{\xi n_1^2 + \bar{\xi}n_2^2}{2(\xi n_1 + \bar{\xi}n_2)}.
	\end{align}
\end{lemma}
The proof can be found in Appendix \ref{app:OptimalStationaryAoI}.
The optimal policy can be found by searching for the optimal parameter $\xi$. Based
on the previous lemma, we can derive a quadratic equation for the optimal $\xi$:

\begin{lemma}\label{lem:OptimalBetaAoI}
	Let $\xi_{\mathrm{avg}}^*$ be the optimal parameter for the stationary randomized policy that minimizes the average \gls{aoi}.
	Let $\tilde{\xi}_{1,2}$ denote the solutions of the following quadratic equation:
	\begin{align}
		 & 2\bigl(n_1\bar{\epsilon}_2-n_2\bar{\epsilon}_1\bigr)\bigl(\xi n_1+\bar{\xi}n_2\bigr)^2 \notag \\
		 & \!+\!
		n_1n_2(n_1-n_2)\bigl(\xi \bar{\epsilon}_1+\bar{\xi}\bar{\epsilon}_2\bigr)^2 \!\!=\! 0.
	\end{align}
	Then the optimal parameter $\xi_{\mathrm{avg}}^*$ can be found by calculating the average \gls{aoi}
	at $\xi = 0$ and $\xi = 1$ and at the solutions $\tilde{\xi}_{1,2}$ that lie in the interval $[0,1]$
	and choosing the $\xi$ that results in the lowest average \gls{aoi}.
\end{lemma}

The proof can be found in Appendix \ref{app:OptimalBetaAoI}.
So far we have only derived a condition under which constant
schedules are suboptimal within the class of
stationary randomized policies.
Next, we show that this also implies suboptimality
within the class of deterministic schedules $\Pi_D$.
The basic step here is to show that
for every stationary randomized policy,
there exists at least one deterministic schedule
with the same or a lower average \gls{aoi}.
\begin{theorem}\label{lem:ConstantSupoptimalityCertificate}
	If the optimal parameter $\xi_{\mathrm{avg}}^*$ defined in Lemma \ref{lem:OptimalBetaAoI} satisfies $\xi_{\mathrm{avg}}^* \in (0,1)$ and
	if the average age at this parameter is strictly smaller than the average age at $\xi = 0$ and $\xi = 1$, then
	the two constant schedules that always choose $n_1$ or $n_2$ are suboptimal within the class of deterministic schedules $\Pi_D$.
\end{theorem}

The proof can be found in Appendix \ref{app:ConstantSupoptimalityCertificate}.

In the next section, similar results for the exponential \gls{aoi} metric are given.
\subsection{Stationary Randomized Policies for Exponential \gls{aoi}}

\begin{lemma}\label{lem:StationaryRandomizedExpAoI}
	Let $\pi \in \Pi_{SR}$ be a stationary randomized policy that chooses $n_1$ with probability $\xi$ and
	$n_2$ with probability $\bar{\xi}$. Further, define
	\begin{align}
		A(\xi)                & \triangleq \xi \bar{\epsilon}_1e^{\alpha n_1} + \bar{\xi}\bar{\epsilon}_2e^{\alpha n_2}, \\
		B(\xi)                & \triangleq \xi \epsilon_1 e^{\alpha n_1} + \bar{\xi}\epsilon_2 e^{\alpha n_2},           \\
		n_{\mathrm{avg}}(\xi) & \triangleq \xi n_1 + \bar{\xi}n_2.
	\end{align}
	If $B(\xi)<1$, then the average exponential age of $\pi$ is given by
	\begin{align}
		\Delta_{\exp,\mathrm{avg}}(\xi)
		=
		\frac{A(\xi)}{\alpha n_{\mathrm{avg}}(\xi)}
		\left(
		\frac{A(\xi)}{1-B(\xi)}-1
		\right).
	\end{align}
\end{lemma}
The proof can be found in Appendix \ref{app:StationaryRandomizedExpAoI}.

\begin{lemma}\label{lem:OptimalBetaExpAoI}
	Let $\xi_{\exp}^*$ be an optimal parameter for a stationary randomized policy minimizing the average exponential age penalty, and let
	$A(\xi)$, $B(\xi)$, and $n_{\mathrm{avg}}(\xi)$ be defined as in Lemma \ref{lem:StationaryRandomizedExpAoI}.
	Then every interior optimizer $\xi_{\exp}^* \in (0,1)$ satisfies
	\begin{align} \label{eq:stationaryConditionCompactExp}
		 & \Bigl(A'(\xi)\bigl(A(\xi)+B(\xi)-1\bigr)+A(\xi)\bigl(A'(\xi)+B'(\xi)\bigr)\Bigr)\, \notag \\
		 & n_{\mathrm{avg}}(\xi)\,(1-B(\xi))
		-
		A(\xi)\bigl(A(\xi)+B(\xi)-1\bigr) \notag                                                     \\
		 & \Bigl(n_{\mathrm{avg}}'(\xi)(1-B(\xi))-n_{\mathrm{avg}}(\xi)B'(\xi)\Bigr)
		=0,
	\end{align}
	which is equivalent to
	a quadratic equation in $\xi$.
	Consequently, $\xi_{\exp}^*$ can be found by
	evaluating $\Delta_{\exp,\mathrm{avg}}(\xi)$ at $\xi=0$, $\xi=1$, and at
	the roots of the equation above
	that lie in $[0,1]$ and satisfy $B(\xi)<1$, and
	then selecting the parameter with the smallest
	average exponential \gls{aoi}.

\end{lemma}
The proof can be found in Appendix \ref{app:OptimalBetaExpAoI}.
We now state a corresponding version of Theorem \ref{lem:ConstantSupoptimalityCertificate} for the exponential \gls{aoi} case.
\begin{theorem}\label{lem:ConstantSupoptimalityCertificateExpAoI}
	Assume that $\rho_1 < 1$ and $\rho_2 < 1$.
	If the optimal parameter $\xi_{\exp}^*$ defined in Lemma \ref{lem:OptimalBetaExpAoI}
	satisfies $\xi_{\exp}^* \in (0,1)$ and
	if the average exponential age at this parameter is strictly smaller than
	the average exponential age at $\xi = 0$ and $\xi = 1$, then
	the two constant schedules that always choose $n_1$ or $n_2$ are suboptimal within the
	class of deterministic schedules $\Pi_D$.
\end{theorem}

The proof can be found in Appendix \ref{app:ConstantSupoptimalityCertificateExpAoI}.

\section{Numerical Results}\label{sec:numericalResults}

In this section, we present numerical
evaluations of the results
to gain insight into age-optimal transmission.

\subsection{Visualization of the Optimality Regions}
We start with evaluations of the optimality regions given by
Theorems \ref{thm:chooseN1OptimalUnified}, \ref{thm:chooseN2OptimalUnified},
\ref{thm:chooseN1OptimalExpUnified}, and \ref{thm:chooseN2OptimalExpUnified}.
In addition, the regions given by Theorem~\ref{lem:ConstantSupoptimalityCertificate}
and \ref{lem:ConstantSupoptimalityCertificateExpAoI} where
a constant schedule is not optimal are also marked in the figures.
The regions are calculated
by numerically evaluating the conditions
in Theorems \ref{thm:chooseN1OptimalUnified}, \ref{thm:chooseN2OptimalUnified},
\ref{thm:chooseN1OptimalExpUnified}, and \ref{thm:chooseN2OptimalExpUnified} for
each parameter tuple with $\epsilon_1 > \epsilon_2$.
The results for
the average age are shown in Fig. \ref{fig:optimalityRegionsAoI}.

\begin{figure*}[!t]
	\centering

	\subfloat{%
		\includegraphics[width=0.45\textwidth]{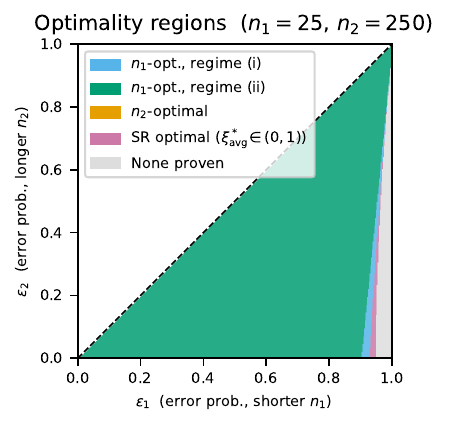}
		\label{fig:1a}
	}\hfil
	\subfloat{%
		\includegraphics[width=0.45\textwidth]{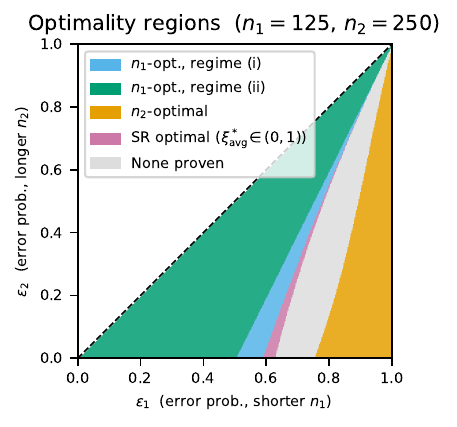}
		\label{fig:1b}
	}

	\medskip

	\subfloat{%
		\includegraphics[width=0.45\textwidth]{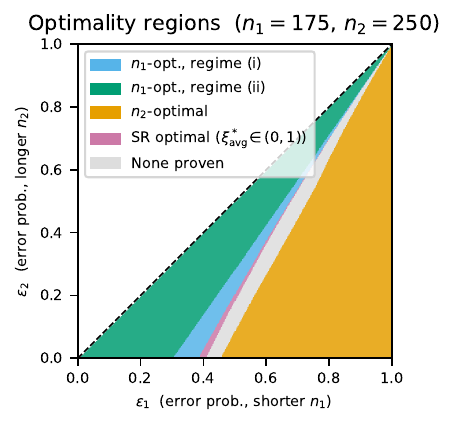}
		\label{fig:1c}
	}\hfil
	\subfloat{%
		\includegraphics[width=0.45\textwidth]{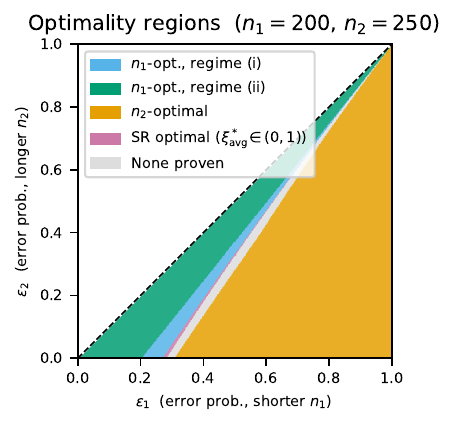}
		\label{fig:1d}
	}

	\caption{Optimality regions for different values of $n_1$ and $n_2$.}
	\label{fig:optimalityRegionsAoI}
\end{figure*}
For the evaluations, we fix the longer blocklength $n_2 = 250$ and vary the shorter blocklength $n_1$.
The theorems cover a large fraction of the parameter space,
although a region
remains where it is not clear if a constant schedule is optimal or not.
Moreover, the region in which
Theorem~\ref{lem:ConstantSupoptimalityCertificate} certifies the suboptimality of constant
schedules is relatively small.
Next, we compare these optimality regions to the corresponding regions for the exponential age criterion.
The results are shown in Fig. \ref{fig:optimalityRegionsExpAoI}.
\begin{figure*}[!t]
	\centering

	\subfloat{%
		\includegraphics[width=0.48\textwidth]{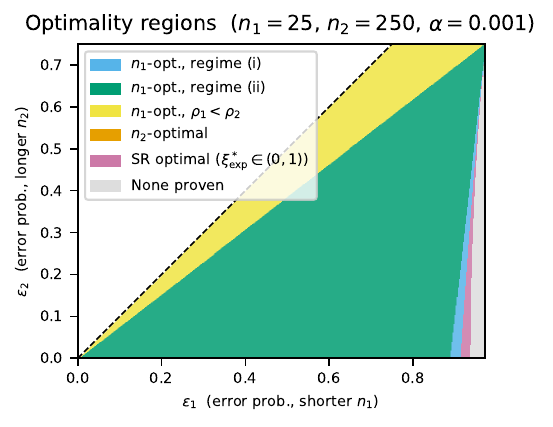}
		\label{fig:2a}
	}\hfil
	\subfloat{%
		\includegraphics[width=0.48\textwidth]{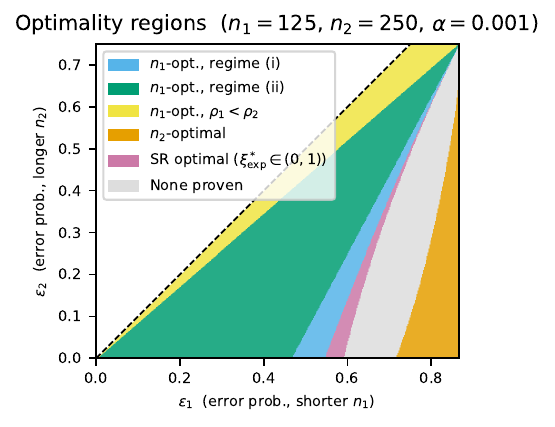}
		\label{fig:2b}
	}

	\medskip

	\subfloat{%
		\includegraphics[width=0.48\textwidth]{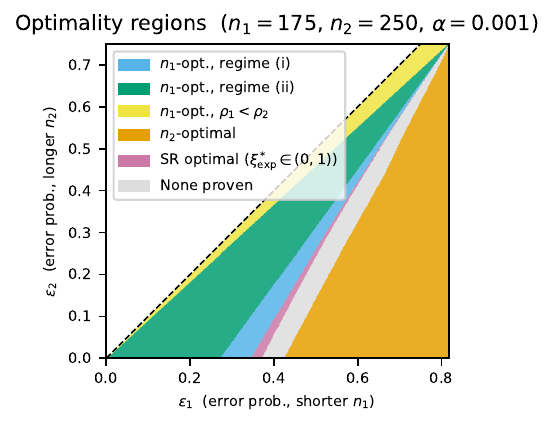}
		\label{fig:2c}
	}\hfil
	\subfloat{%
		\includegraphics[width=0.48\textwidth]{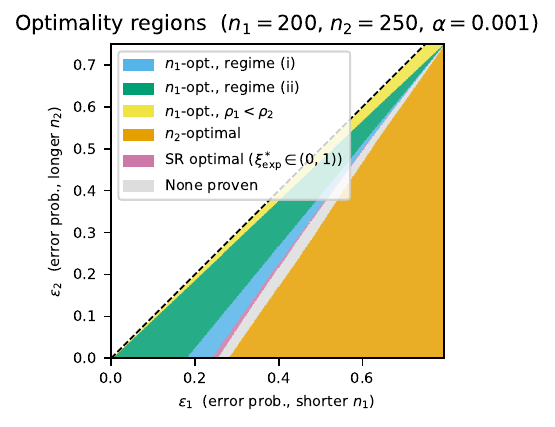}
		\label{fig:2d}
	}

	\caption{Optimality regions for the exponential age criterion for
		different values of $n_1$ and $n_2$. The view is restricted to the region
		where $\rho_1 < 1$ and $\rho_2 < 1$.}
	\label{fig:optimalityRegionsExpAoI}
\end{figure*}
The general behavior of the regions is similar to that in the average \gls{aoi} case.
Interestingly, for both metrics, the region where none of the theorems apply is the largest
for the blocklengths $n_1 = 125, n_2 = 250$.
A notable observation is that the boundary of the certified $n_1$-optimality region nearly coincides
with the boundary of the region in which the constant $n_1$ schedule
is provably outperformed by non-constant schedules.
This suggests that the sufficient conditions in
Theorems \ref{thm:chooseN1OptimalUnified} and \ref{thm:chooseN1OptimalExpUnified}
may also be necessary conditions
for the optimality of a constant schedule with the shorter blocklength.
To see this, note that the pink region appears to separate the gray region
from the region in which the constant $n_1$ schedule is optimal.
Moreover, moving to the right in the plot corresponds to increasing $\epsilon_1$.
Since a larger $\epsilon_1$ makes transmissions with blocklength $n_1$ less
reliable, the constant $n_1$ schedule cannot become optimal again further to the right.
The two-blocklength results illustrate the structure of the optimality regions,
but in practical applications, there are often more than two blocklengths available.
Therefore, in the next subsection, we give a multi-blocklength example based on the
extension described in Section~\ref{sec:MultiBlocklengthOptimalityConditions}.

\subsection{Example: Coding over a \gls{bsc}}
We consider the problem of transmitting a fixed payload
of $k=200$ bits over a \gls{bsc} with
a raw bit error probability of $p=0.05$ so that
the average age is minimized.
We use the physical layer module of the Sionna library \cite{sionna} to
simulate a practical system.
The block error probability was estimated for different blocklengths of the 5GLDPCCode.
The 5GLDPCCode from Sionna is consistent with the 5G LDPC code defined in 3GPP TS 38.212.
The number of iterations for the decoder was set to 20, and the number
of transmissions per blocklength to estimate the error probabilities was set to $200 000$. We simulated the
blocklengths $\{325, 350, 375, 400, 425\}$.
The corresponding error probabilities as well as the application of the method described in
Sec. \ref{sec:MultiBlocklengthOptimalityConditions} are summarized in Table \ref{tab:multi_blocklength_optimality}.
\begin{table*}[!t]
	\centering
	\renewcommand{\arraystretch}{1.25}
	\setlength{\tabcolsep}{4pt}

	\caption{Multi-blocklength optimality check. Global bounds:
		$L = 533.3333$, $U = 551.3441$. Globally optimal blocklength: $350$.}
	\label{tab:multi_blocklength_optimality}

	\scriptsize
	\begin{tabular}{
		|>{\centering\arraybackslash}m{2.3cm}
		|>{\centering\arraybackslash}m{1.55cm}
		|>{\centering\arraybackslash}m{1.55cm}
		|>{\centering\arraybackslash}m{1.55cm}
		|>{\centering\arraybackslash}m{1.55cm}
		|>{\centering\arraybackslash}m{1.55cm}
		|>{\centering\arraybackslash}m{1.55cm}|}
		\hline
		 &
		\makecell{$n_1 = 325$     \\ $\epsilon = 0.22$}
		 &
		\makecell{$n_2 = 350$     \\ $\epsilon = 0.07$}
		 &
		\makecell{$n_3 = 375$     \\ $\epsilon = 0.016$}
		 &
		\makecell{$n_4 = 400$     \\ $\epsilon = 0.004$}
		 &
		\makecell{$n_5 = 425$     \\ $\epsilon = 0.001$}
		 &
		\makecell{\textbf{Global} \\ \textbf{optimal?}}
		\\
		\hline

		\makecell{$n^* = 325$,    \\ $\epsilon^* = 0.22$}
		 &
		\cellcolor{graycell}--
		 &
		\cellcolor{redcell}
		 &
		\cellcolor{redcell}
		 &
		\cellcolor{bluecell}(i)
		 &
		\cellcolor{greencell}(ii)
		 &
		\cellcolor{graycell}
		\\
		\hline

		\makecell{$n^* = 350$,    \\ $\epsilon^* = 0.07$}
		 &
		\cellcolor{goldcell}Thm.4
		 &
		\cellcolor{graycell}--
		 &
		\cellcolor{greencell}(ii)
		 &
		\cellcolor{greencell}(ii)
		 &
		\cellcolor{greencell}(ii)
		 &
		\cellcolor{lightgreen}{$\star$ opt.}
		\\
		\hline

		\makecell{$n^* = 375$,    \\ $\epsilon^* = 0.016$}
		 &
		\cellcolor{goldcell}Thm.4
		 &
		\cellcolor{redcell}
		 &
		\cellcolor{graycell}--
		 &
		\cellcolor{greencell}(ii)
		 &
		\cellcolor{greencell}(ii)
		 &
		\cellcolor{graycell}
		\\
		\hline

		\makecell{$n^* = 400$,    \\ $\epsilon^* = 0.004$}
		 &
		\cellcolor{redcell}
		 &
		\cellcolor{redcell}
		 &
		\cellcolor{redcell}
		 &
		\cellcolor{graycell}--
		 &
		\cellcolor{greencell}(ii)
		 &
		\cellcolor{graycell}
		\\
		\hline

		\makecell{$n^* = 425$,    \\ $\epsilon^* = 0.001$}
		 &
		\cellcolor{redcell}
		 &
		\cellcolor{redcell}
		 &
		\cellcolor{redcell}
		 &
		\cellcolor{redcell}
		 &
		\cellcolor{graycell}--
		 &
		\cellcolor{graycell}
		\\
		\hline
	\end{tabular}
\end{table*}
In each row, the blocklength on the left is compared with all the other blocklengths.
The bounds $L$ and $U$ are computed as described in Section \ref{sec:MultiBlocklengthOptimalityConditions}.
The red cells indicate that for the corresponding pairwise comparison
no strictly positive age gap can be guaranteed. The orange, blue, and green cells indicate
which theorem and parameter regime is applied to guarantee a strictly positive age gap.
The overall result is that for this set of possible blocklengths with the corresponding
error probabilities, it is optimal to always transmit with the blocklength $n_2 = 350$ with respect
to the time-average age criterion.
A similar evaluation can be performed for the exponential age criterion, as shown
in Tab. \ref{tab:multi_blocklength_optimality_exp_aoi}.
\begin{table*}[!t]
	\centering
	\renewcommand{\arraystretch}{1.25}
	\setlength{\tabcolsep}{4pt}

	\caption{Multi-blocklength optimality check for exponential AoI.
		$\alpha = 0.0050$, $L_{\exp} = 17.0180$, $U_{\exp} = 21.0937$.
		Globally optimal blocklength: $375$.}
	\label{tab:multi_blocklength_optimality_exp_aoi}

	\scriptsize
	\begin{tabular}{
		|>{\centering\arraybackslash}m{2.3cm}
		|>{\centering\arraybackslash}m{1.65cm}
		|>{\centering\arraybackslash}m{1.65cm}
		|>{\centering\arraybackslash}m{1.65cm}
		|>{\centering\arraybackslash}m{1.65cm}
		|>{\centering\arraybackslash}m{1.65cm}|}
		\hline
		 &
		\makecell{$n_1 = 350$     \\ $\epsilon = 0.07$}
		 &
		\makecell{$n_2 = 375$     \\ $\epsilon = 0.016$}
		 &
		\makecell{$n_3 = 400$     \\ $\epsilon = 0.004$}
		 &
		\makecell{$n_4 = 425$     \\ $\epsilon = 0.001$}
		 &
		\makecell{\textbf{Global} \\ \textbf{optimal?}}
		\\
		\hline

		\makecell{$n^* = 350$,    \\ $\epsilon^* = 0.07$}
		 &
		\cellcolor{graycell}--
		 &
		\cellcolor{redcell}
		 &
		\cellcolor{redcell}
		 &
		\cellcolor{redcell}
		 &
		\cellcolor{graycell}
		\\
		\hline

		\makecell{$n^* = 375$,    \\ $\epsilon^* = 0.016$}
		 &
		\cellcolor{goldcell}Thm.6
		 &
		\cellcolor{graycell}--
		 &
		\cellcolor{greencell}(ii)
		 &
		\cellcolor{greencell}(ii)
		 &
		\cellcolor{lightgreen}{$\star$ opt.}
		\\
		\hline

		\makecell{$n^* = 400$,    \\ $\epsilon^* = 0.004$}
		 &
		\cellcolor{redcell}
		 &
		\cellcolor{redcell}
		 &
		\cellcolor{graycell}--
		 &
		\cellcolor{greencell}(ii)
		 &
		\cellcolor{graycell}
		\\
		\hline

		\makecell{$n^* = 425$,    \\ $\epsilon^* = 0.001$}
		 &
		\cellcolor{redcell}
		 &
		\cellcolor{redcell}
		 &
		\cellcolor{redcell}
		 &
		\cellcolor{graycell}--
		 &
		\cellcolor{graycell}
		\\
		\hline
	\end{tabular}
\end{table*}
In this case, however,
we observed that the applicability of the theorems to guarantee the optimal transmission scheme
is reduced. For example, when we apply the method to the set of blocklengths
and error probabilities in Tab.~\ref{tab:multi_blocklength_optimality}, but with
the exponential age criterion with varying parameters $\alpha$,
no blocklength can be guaranteed to be optimal.
In general, our numerical evaluations indicate that as
the number of tuples $(n_i, \epsilon_i)$ increases, the
portion of the parameter space where the method from Section \ref{sec:MultiBlocklengthOptimalityConditions}
can be applied to guarantee
the optimality of a constant schedule decreases. This is expected from
the theorems
because the number of pairwise comparison conditions that need to be satisfied
increases.

\section{Conclusion}\label{sec:conclusion}

In this work, we investigated status updating without
feedback when the
transmitter can choose from a discrete
set of blocklength error probability pairs.
We first derived closed-form expressions
for the average and average exponential \gls{aoi} for
deterministic periodic schedules.
Interestingly, our numerical evaluations of these expressions
revealed that it is generally not optimal to use
a constant blocklength and that the gains
achievable by more sophisticated schedules
with different blocklengths can be significant
for the average exponential \gls{aoi} metric.
Motivated by these observations,
our main contribution is the derivation of
easily checkable sufficient conditions
under which a constant schedule is optimal.
Our second main contribution is the identification
of parameter regimes where a non-constant
deterministic schedule
outperforms the best constant schedule.
Numerical evaluations of these conditions
show that, for a large part of the parameter space,
the constant schedules are optimal.
The results also indicate that the optimality conditions
for a constant schedule with the shorter blocklength are not
only sufficient but might also be necessary.
Extensions of the optimality conditions for the case
of $K > 2$ blocklengths were also given.
A practical example of provable \gls{aoi} optimal transmission
over a \gls{bsc} with different blocklengths of a 5G LDPC code was also provided,
which demonstrated the applicability of the results.
The practical implications of this work are twofold.
First, for a large part of the parameter space,
it is optimal to use a constant blocklength.
Moreover, for the \gls{aoi} metric,
the achievable gains in the regions where
a constant schedule is not optimal are quite small.

Nevertheless, this work also demonstrates that
this is not the case for other age-related
metrics, such as the average exponential
\gls{aoi}. The second insight is
therefore that, for this metric, the parameter
regions where a constant schedule is not
optimal are also quite small, but
the achievable gains in these regions can be significant.

\appendices
\section{Proof of Theorem \ref{theo:AoICyclicClosedForm}}\label{app:AoICyclicClosedForm}

\begin{proof}
	First, note that the age process has the form of a Markov renewal process.
	Define state $i$ as the state in which the next transmission
	after a successful reception uses blocklength $s_i$.
	Let $Y_i$ denote the corresponding inter-reception
	interval until the next successful transmission.
	By the Markov renewal reward theorem, the average \gls{aoi} is given by
	% A good reference for this could be "Markov Renewal Theory: A Survey"
	\begin{align}\label{eq:MarkovRenewalReward}
		\Delta_{\mathrm{avg}}
		=
		\frac{\sum_{i=0}^{P-1} \nu_i\mathbb{E}[Q_i]}{\sum_{i=0}^{P-1} \nu_i\mathbb{E}[Y_i]},
	\end{align}
	where $Q_i$ is the accumulated area under the
	\gls{aoi} curve during interval $Y_i$ and
	$\nu_i$ is the stationary probability of being in state $i$.
	All schedule-position indices in this proof are interpreted modulo $P$.
	If an inter-reception interval starts in state $i$,
	then the AoI at the beginning of the interval equals the blocklength
	of the previous transmission in the schedule.
	The area $Q_i$ accumulated during $Y_i$ that starts in state $i$ and
	its expected value are given by
	\begin{align}
		Q_i
		  & =
		\int_0^{Y_i}
		\left(s_{i-1}+t\right)\,dt
		=
		s_{i-1}Y_i+\frac{Y_i^2}{2}, \\
		\mathbb{E}[Q_i]
		= &
		s_{i-1}\mathbb{E}[Y_i]
		+
		\frac{\mathbb{E}[Y_i^2]}{2}.
		\label{eq:EQi_general}
	\end{align}
	Next, we need expressions for $\mathbb{E}[Y_i]$ and $\mathbb{E}[Y_i^2]$.
	Conditioning on the first transmission outcome gives
	\begin{align}
		\mathbb{E}[Y_i]
		 & =
		p_i s_i
		+
		q_i\left(s_i+\mathbb{E}[Y_{i+1}]\right) \notag \\
		 & =
		s_i+q_i\mathbb{E}[Y_{i+1}],
		\qquad i=0,\ldots,P-1.
		\label{eq:EYi_recursion_average}
	\end{align}
	Iterating~\eqref{eq:EYi_recursion_average} over one full period yields
	\begin{align}
		 & \mathbb{E}[Y_i]
		=
		s_i
		+
		q_i s_{i+1}
		+
		q_i q_{i+1} s_{i+2}
		+
		\cdots \notag      \\
		 & \quad
		+
		\left(\prod_{r=0}^{P-2}q_{i+r}\right)s_{i+P-1}
		+
		\left(\prod_{r=0}^{P-1}q_{i+r}\right)
		\mathbb{E}[Y_i].
	\end{align}
	Solving for $\mathbb{E}[Y_i]$ gives (\ref{eq:EY}) in the theorem,
	where the empty product $\prod_{r=0}^{-1}(\cdot)$ is defined as $1$.
	This expression is finite whenever $\prod_{j=0}^{P-1} q_j < 1$.
	For $\mathbb{E}[Y_i^2]$, we again condition on the first transmission outcome to get
	\begin{align}
		 & \mathbb{E}[Y_i^2]
		=
		p_i \bigl(s_i\bigr)^2
		+
		q_i\mathbb{E}\!\left[
			\left(s_i+Y_{i+1}\right)^2
		\right] \notag       \\
		 & =
		p_i \bigl(s_i\bigr)^2
		+
		q_i
		\left(
		\bigl(s_i\bigr)^2
		+
		2s_i\mathbb{E}[Y_{i+1}]
		+
		\mathbb{E}[Y_{i+1}^2]
		\right) \notag       \\
		 & =
		\bigl(s_i\bigr)^2
		+
		q_i
		\left(
		2s_i\mathbb{E}[Y_{i+1}]
		+
		\mathbb{E}[Y_{i+1}^2]
		\right),
		\label{eq:EYisq_recursion_average}
	\end{align}
	for $i=0,\ldots,P-1$.
	Iterating~\eqref{eq:EYisq_recursion_average} over one full period gives
	\begin{align}
		\mathbb{E}[Y_i^2]
		 & =
		\sum_{k=0}^{P-1}
		\left(\prod_{r=0}^{k-1}q_{i+r}\right)
		\bigl(s_{i+k}\bigr)^2 \notag \\
		 & \quad
		+
		2
		\sum_{k=0}^{P-1}
		\left(\prod_{r=0}^{k}q_{i+r}\right)
		s_{i+k}
		\mathbb{E}[Y_{i+k+1}] \notag \\
		 & \quad
		+
		\left(\prod_{r=0}^{P-1}q_{i+r}\right)
		\mathbb{E}[Y_i^2].
	\end{align}
	Rearranging yields (\ref{eq:EYSquared}) in the theorem.
	Again, finiteness is ensured whenever $\prod_{j=0}^{P-1} q_j<1$.

	The only missing part is the stationary probabilities
	$\nu_i$. These states form a finite, aperiodic,
	irreducible Markov chain.
	By calculating the transition probabilities
	and solving the corresponding balance equations,
	the stationary distribution is given as
	$\nu_i
		=
		\frac{p_{i-1}}{\sum_{j=0}^{P-1} p_j},
		\text{ for } i=0,\ldots,P-1$.
	Inserting everything into~\eqref{eq:MarkovRenewalReward}
	gives the closed-form expression for the average \gls{aoi}.
\end{proof}

\section{Proof of Theorem \ref{theo:ExponentialAoICyclicClosedForm}}\label{app:ExponentialAoICyclicClosedForm}

\begin{proof}
	The proof follows the same steps as the proof
	of Theorem \ref{theo:AoICyclicClosedForm}.
	The only difference is the calculation
	of the area under the exponential \gls{aoi} curve
	during one inter-reception interval $Y_i$.
	This area $Q_i^{\exp}$ and its expected
	value can be
	calculated as follows:
	\begin{align}
		Q_i^{\exp}
		= & \int_0^{Y_i} e^{\alpha(s_{i-1}+t)}\,dt
		=\frac{e^{\alpha s_{i-1}}}{\alpha}\left(e^{\alpha Y_i}-1\right), \\
		\mathbb{E}[Q_i^{\exp}]
		= &
		\frac{e^{\alpha s_{i-1}}}{\alpha}\left(M_i-1\right),
		\quad\!\! \text{with}\quad\!\!
		M_i  \triangleq \mathbb{E}\!\left[e^{\alpha Y_i}\right].
	\end{align}
	The derivation for $\mathbb{E}[Y_i]$ is the same as in the proof of Theorem \ref{theo:AoICyclicClosedForm}.
	For the $M_i$ terms, conditioning on the outcome of the first transmission in state $i$ yields
	\begin{align}
		M_i
		 & = p_i e^{\alpha s_i} + q_i e^{\alpha s_i}\, M_{i+1},
		\quad i=0,\ldots,P-1.
		\label{eq:Mi_recursion_general}
	\end{align}
	With $\rho_i \triangleq q_i e^{\alpha s_i}$ and $\gamma_i \triangleq p_i e^{\alpha s_i}$,
	\eqref{eq:Mi_recursion_general}
	can be written as
	\begin{align}
		M_i = \gamma_i + \rho_i M_{i+1}.
		\label{eq:Mi_affine}
	\end{align}
	Iterating~\eqref{eq:Mi_affine} over one full period gives
	\begin{align}
		M_i
		 & = \gamma_i + \rho_i \gamma_{i+1} + \rho_i \rho_{i+1} \gamma_{i+2} + \cdots \\
		 & + \left(\prod_{r=0}^{P-2} \rho_{i+r}\right)\gamma_{i+P-1}
		+ \left(\prod_{r=0}^{P-1} \rho_{i+r}\right) M_i.
	\end{align}
	Rearranging yields the closed form
	\begin{align}
		M_i
		=
		\frac{
			\displaystyle \sum_{k=0}^{P-1}
			\left(\prod_{r=0}^{k-1} \rho_{i+r}\right)\, \gamma_{i+k}
		}{
			\displaystyle 1-\prod_{r=0}^{P-1} \rho_{i+r}
		},
		\qquad i=0,\ldots,P-1,
		\label{eq:Mi_closed_form_general}
	\end{align}
	where the empty product $\prod_{r=0}^{-1}(\cdot)$ is defined as $1$.
	Since $\rho_i\ge 0$, finiteness of $M_i$ is ensured whenever
	\begin{align}
		\prod_{j=0}^{P-1} \rho_j
		=
		e^{\alpha\sum_{j=0}^{P-1} s_j}\prod_{j=0}^{P-1} q_j
		<1.
		\label{eq:Mi_finite_condition}
	\end{align}

	The stationary distribution is the same as in the
	proof of Theorem \ref{theo:AoICyclicClosedForm}.
	Combining everything gives the closed-form
	expression.
\end{proof}

\section{Proof of Lemma \ref{lem:DinkelbachTrafo}}\label{app:DinkelbachTrafo}

\begin{proof}
	The original objective $\Delta_{\mathrm{avg}}(\pi)$ in (\ref{eq:gObjective}) can be expressed as
	\begin{align}
		 & \Delta_{\mathrm{avg}}(\pi) = \limsup_{T\to\infty} \frac{C_T^{\pi}}{A_T^{\pi}}, \\
		 & \text{where } C_T^{\pi}
		= \sum_{i=0}^{T} \tilde{c}(b_i,a_i),\quad A_T^{\pi} = \sum_{i=0}^{T} a_i.
	\end{align}
	Now, consider the parametrized objective
	\begin{align}
		D_T^{\pi}(\lambda) = C_T^{\pi} - \lambda A_T^{\pi} = \sum_{i=0}^{T} (\tilde{c}(b_i,a_i) - \lambda a_i).
	\end{align}
	Note that the following holds for all $\lambda$:
	\begin{align}
		\frac{C_T^{\pi}}{A_T^{\pi}} - \lambda = \frac{D_T^{\pi}(\lambda)}{A_T^{\pi}}.
	\end{align}
	Taking the $\limsup$ on both sides gives
	\begin{align}
		\Delta_{\mathrm{avg}}(\pi) \leq \lambda \Leftrightarrow \limsup_{T\to\infty} \frac{D_T^{\pi}(\lambda)}{A_T^{\pi}} \leq 0
	\end{align}
	Using this, we can express our original optimization problem as follows:
	\begin{align}
		 & \inf_{\pi\in \Pi_D} \limsup_{T\to\infty} \frac{\sum_{i=0}^{T}\tilde{c}(b_i,a_i)}{\sum_{i=0}^{T}a_i}                  \\
		 & = \inf \left\{ \lambda: \inf_{\pi\in \Pi_D} \limsup_{T\to\infty} \frac{D_T^{\pi}(\lambda)}{A_T^{\pi}} \leq 0\right\} \\
		 & = \inf \left\{ \lambda: \inf_{\pi\in \Pi_D} \limsup_{T\to\infty} \frac{1}{T} \sum_{i=0}^{T} (\tilde{c}(b_i,a_i)
		- \lambda a_i) \leq 0\right\},\label{eq:NormalAverageCostTransform}
	\end{align}
	where the last equality holds because for every schedule $\pi \in \Pi_D$ and every $T \geq 0$,
	$	A_T^\pi = \sum_{i=0}^T a_i$,
	with $a_i \in \{n_1,n_2\}$, and hence
	\begin{align}
		(T+1)n_1 \leq A_T^\pi \leq (T+1)n_2.
	\end{align}
	We now write the argument of the limsup as
	\begin{align}
		\frac{D_T^\pi(\lambda)}{A_T^\pi}
		=
		\frac{D_T^\pi(\lambda)}{T+1}\cdot \frac{T+1}{A_T^\pi}.
	\end{align}
	Since the second factor on the RHS
	is uniformly bounded above and
	below by positive constants as
	$\frac{1}{n_2} \leq \frac{T+1}{A_T^\pi} \leq \frac{1}{n_1}$,
	it follows by standard properties of the $\limsup$
	that
	\begin{align}
		\limsup_{T\to\infty}\frac{D_T^\pi(\lambda)}{A_T^\pi}\leq 0
		\quad \Longleftrightarrow \quad
		\limsup_{T\to\infty}\frac{D_T^\pi(\lambda)}{T}\leq 0.
	\end{align}
\end{proof}

\section{Proof of Lemma \ref{lem:boundedcostLemmaAoI}} \label{app:boundedcostLemmaAoI}

\begin{proof}

	We start with the expression for the
	expected age $m_i$, which is given by
	\begin{align}
		m_i = \sum_{x\in \mathcal{S}} b_i(x)x = \mathbb{E}[X_i],
	\end{align}
	where $X_i$ denotes the age at the beginning of the $i$-th transmission.
	Next, we derive a recursion for $m_i$. Depending on action $a_i$, we have
	\begin{align}
		X_{i+1} =
		\begin{cases}
			a_i,       & \text{with probability } \bar{\epsilon}(a_i), \\
			X_i + a_i, & \text{with probability } \epsilon(a_i).
		\end{cases}
	\end{align}
	Taking expectations yields
	\begin{align}\label{eq:RecursionForMk}
		m_{i+1}
		 & = \bar{\epsilon}(a_i)a_i
		+ \epsilon(a_i)\big(m_i + a_i\big) \\
		 & = a_i + \epsilon(a_i)m_i.
	\end{align}
	Define the map $f_e(m) = n_e + \epsilon_e m$ for $e \in \{1,2\}$.
	Note that both maps have a unique fixed point, which is given by $n_e/\bar{\epsilon}_e$.
	Since $m_0=n_1\leq m_U$ and
	$f_e(m_U)\leq m_U$ for both actions, induction yields $m_i\leq m_U$ for all $i$.
	Moreover,
	\begin{align}
		\max\{m_L\!-m_{i+1},0\}
		\!\leq \epsilon(a_i)\max\{m_L-m_i,0\}.
	\end{align}
	Thus, with $\epsilon_{\max}\triangleq\max\{\epsilon_1,\epsilon_2\}<1$,
	the distance below $m_L$ is at most
	$\epsilon_{\max}^i\max\{m_L-m_0,0\}$ and hence converges to zero uniformly over
	all action sequences. This proves the lower bound $m_i\geq m_L-\delta$
	for all sufficiently large $i$.
\end{proof}

\section{Proof of Lemma \ref{lem:AoIBounds}}\label{app:AoIBounds}
\begin{proof}
	The upper bound is obtained by taking the minimum of the two constant schedules that always
	choose \(n_1\) or always choose \(n_2\).
	The lower bound is the combination of two bounds: The first bound is the average
	\gls{aoi} of a blocklength error probability tuple that is given by
	the shortest blocklength and lowest error probability available in $\mathcal{C}$.
	We now prove the second
	lower bound, which is inspired
	by \cite{kadotaOptimizingAgeInformation2018}.
	Let $Y_0, Y_1, Y_2, \ldots$ be the time intervals between successful transmissions.
	By the initialization in Section~\ref{sec:system_model}, an update of blocklength
	$n_1$ is successfully received at time $t=0$, and the initial age is $n_1$.
	Let $T'$ be a fixed time horizon.
	Over interval \(i\), let \(a_{i-}\) denote the blocklength of the update
	received at the beginning of the interval.
	In particular, $a_{0-}=n_1$.
	The area under the AoI curve during the interval
	$Y_i$ is
	\begin{align}
		Q_i=a_{i-}Y_i+\frac{Y_i^2}{2}.
	\end{align}
	Therefore, the expected average \gls{aoi} under
	any admissible schedule can be lower-bounded as
	\begin{align}
		\Delta_{\mathrm{avg}}(\pi)
		  & \geq
		\liminf_{T'\to\infty} \Bigg(
		\frac{1}{T'} \mathbb{E}\left[\sum_{i=0}^{K_{T'}-1}
		a_{i-}Y_i + a_{K_{T'}}R_{T'}\right]                                                              \\
		+ & \frac{1}{T'} \mathbb{E}\left[\sum_{i=0}^{K_{T'}-1}\frac{ Y_i^2}{2}+\frac{R_{T'}^2}{2}\right]
		\Bigg) \coloneq \liminf_{T'\to\infty} Z_{T'},
	\end{align}
	where $R_{T'}$ denotes the remaining time interval of $T'$
	after the last successful transmission,
	$K_{T'}$ denotes the number of successful transmissions
	and $a_{K_{T'}}$ denotes the blocklength used at
	the beginning of the remainder interval $R_{T'}$.
	Because $Z_{T'}$ is bounded, we can choose a discrete
	subsequence $T'_m$, so that
	the $\liminf$ of $Z_{T'}$ is attained:
	\begin{align}
		Z_{T'_m}\to \liminf_{T'\to\infty}Z_{T'}.
	\end{align}
	Now, we apply the Bolzano-Weierstrass theorem to extract a further subsequence for which the
	expected success rates, which we denote by $\theta_i$, converge:
	\begin{align}
		\lim_{T'_m \to\infty} \mathbb{E}\left[\frac{S_i(T'_m)}{T'_m}\right] = \theta_i, \qquad i\in\{1,2\},
	\end{align}
	where $S_i(T'_m)$ denotes the number of successful transmissions
	with blocklength $n_i$ up to time $T'_m$.
	This is valid because the expected success rates are
	bounded in $[0,1/n_1]$.
	Note that we have not relabeled the subsequence to ease the notation.
	Now, we consider the terms of $Z_{T'_m}$ separately. For the first part, we have
	\begin{align}
		 & \frac{1}{T'_m} \mathbb{E}\left[\sum_{i=0}^{K_{T'_m}-1} a_{i-}Y_i + a_{K_{T'_m}}R_{T'_m}\right]                   \\
		 & \geq n_1 + (n_2-n_1) \frac{1}{T'_m}
		\mathbb{E}\left[\sum_{i=0}^{K_{T'_m}-1}\mathbf{1}_{\{a_{i-}=n_2\}}Y_i\right]                                        \\
		 & \geq n_1 + (n_2-n_1)n_1\frac{1}{T'_m}\mathbb{E}\left[ \sum_{i=0}^{K_{T'_m}-1}\mathbf{1}_{\{a_{i-}=n_2\}} \right] \\
		 & \xrightarrow{T'_m\to\infty} n_1 + (n_2-n_1)n_1 \theta_2.
	\end{align}
	Now, we consider the second part of $Z_{T'_m}$ and lower-bound it as follows:
	\begin{align}
		 & \frac{1}{T'_m} \mathbb{E}\left[\sum_{i=0}^{K_{T'_m}-1}\frac{ Y_i^2}{2}+\frac{R_{T'_m}^2}{2}\right]                                      \\
		 & = \mathbb{E}\left[ \frac{K_{T'_m}+1}{T'_m} \frac{1}{K_{T'_m}+1}\frac{1}{2}\left(\sum_{i=0}^{K_{T'_m}-1} Y_i^2 +R_{T'_m}^2\right)\right] \\
		 & \geq \mathbb{E}\left[ \frac{1}{2} \frac{K_{T'_m}+1}{T'_m} \left(\frac{T'_m}{K_{T'_m}+1}\right)^2\right]                                 \\
		 & \geq \frac{1}{2}\frac{T'_m}{\mathbb{E}\left[K_{T'_m}+1\right]}
		\xrightarrow{T'_m \to \infty} \frac{1}{2}\frac{1}{\theta_1+\theta_2}.
	\end{align}
	Here we have used Jensen's inequality twice, first to pull the square out of the sample average and
	then to pull the expectation value inside the function $1/(1+x)$, which is convex in $x$.
	The only missing step is now to derive a condition for the rate pair $(\theta_1,\theta_2)$ to be feasible.
	Note that for an arbitrary schedule $\pi$ we have
	\begin{align}
		n_1 N_1(T'_m) + n_2 N_2(T'_m) \leq T'_m,
	\end{align}
	where $N_i(T'_m)$ denotes the number of completed transmissions with blocklength
	$n_i$ up to time $T'_m$.
	With this, we can write
	\begin{align}
		 & \frac{n_1}{\bar{\epsilon}_1}\frac{\bar{\epsilon}_1N_1(T'_m)}{T'_m} +
		\frac{n_2}{\bar{\epsilon}_2}\frac{\bar{\epsilon}_2N_2(T'_m)}{T'_m}
		\leq 1.
	\end{align}
	Note that the second factors in the two products are the
	success rates of the two blocklengths.
	Taking expectation and the limit
	$T'_m \to \infty$ yields the rate pair condition in the theorem.
\end{proof}

\section{Proof of Theorem \ref{thm:chooseN1OptimalUnified}}\label{app:chooseN1OptimalUnified}
\begin{proof}

	We start with the objective formulated in (\ref{eq:DeterministicAoIMDP}).
	First, we show that every finite prefix can be ignored for the average age criterion,
	which allows us to use the bounds in Lemma \ref{lem:boundedcostLemmaAoI} with
	an arbitrarily small $\delta > 0$.
	Let $t^* > 0$ be an arbitrary transmission index. The objective $\Delta_{\mathrm{avg}}(\pi)$ can be written as follows:
	\begin{align}
		 & \Delta_{\mathrm{avg}}(\pi)  = \limsup_{T\to\infty} \left( \frac{\sum_{i=0}^{T}\tilde{c}(b_i,a_i)}{\sum_{i=0}^{T}a_i} \right)                                       \\
		 & = \limsup_{T\to\infty} \left(\frac{\sum_{i=0}^{t^*}\tilde{c}(b_i,a_i)}{\sum_{i=0}^{T}a_i} +\frac{ \sum_{i=t^*+1}^{T}\tilde{c}(b_i,a_i)}{ \sum_{i=0}^{T}a_i}\right) \\
		 & = \limsup_{T\to\infty} \left(\frac{ \sum_{i=t^*+1}^{T}\tilde{c}(b_i,a_i)}{ \sum_{i=0}^{T}a_i} \right)\label{eq:limsupPlusConstant}                                 \\
		 & = \limsup_{T\to\infty} \left( \frac{ \sum_{i=t^*+1}^{T}\tilde{c}(b_i,a_i)}{ \sum_{i=t^*+1}^{T}a_i}\right).\label{eq:anDividedcPlusbn}
	\end{align}
	In \eqref{eq:limsupPlusConstant} we have used that the first summand converges to zero,
	and in \eqref{eq:anDividedcPlusbn} we have used
	that $\limsup_{n\to\infty}\left( \frac{x_n}{c+y_n} \right) = \limsup_{n\to\infty}\left( \frac{x_n}{y_n}\cdot \frac{1}{1+c/y_n}\right) =
		\limsup_{n\to\infty}\left( \frac{x_n}{y_n}\right) \cdot 1$.
	Now, we show that for $\lambda^*$ bounded as in Lemma
	\ref{lem:AoIBounds} and if the conditions
	of the theorem are satisfied, the constant schedule
	that always chooses $n_1$ achieves the minimal
	average cost on the tree in Fig. \ref{fig:TransientMDPChain}.
	To show this, we use the following setup: Let $\pi^{(1)}$ denote the schedule
	that always chooses $n_1$ and let $\tilde{\pi} \in \Pi_D$ denote an arbitrary different schedule.
	Fix a transmission horizon $T$ and a cutoff transmission index
	$T_{\mathrm{c}}\coloneqq T-M<T$ for some fixed integer $M>0$.
	The tail section from $T_{\mathrm{c}}$ to $T$ is of length $M$ and allows us to use geometric bounds
	on the future cost variation caused by changing a single action in a schedule.
	Note that the finiteness of $M$ introduces an error term that we denote
	by $\delta_M > 0$ in the later calculations.
	Let the time indices where the two schedules differ be denoted by
	$\tau_1 < \tau_2 < \ldots < \tau_{N_T}$.
	Since \(\pi^{(1)}\) always chooses \(n_1\), we have
	\begin{align}
		a^{\pi^{(1)}}_{\tau_j} = n_1,
		\qquad
		a^{\tilde{\pi}}_{\tau_j} = n_2,
		\qquad j=1,\ldots,N_{T}.
	\end{align}
	Next, we define a sequence of schedules $\eta^{(0)},\eta^{(1)},\ldots,\eta^{(N_{T})}$
	as follows:
	let $\eta^{(0)}=\tilde{\pi}$, and for $j=1,\ldots,N_T$, let $\eta^{(j)}$
	be obtained from $\eta^{(j-1)}$ by replacing the action at
	time $\tau_{N_T-j+1}$ from $n_2$ to $n_1$. In particular,
	$\eta^{(N_T)}=\pi^{(1)}$.
	Equivalently, for \(j=1,2,\ldots,N_{T}\),
	\begin{align}
		a_i^{\eta^{(j)}} =
		\begin{cases}
			a_i^{\pi^{(1)}},   & \tau_{N_{T}-j+1} \leq i \leq T, \\
			a_i^{\tilde{\pi}}, & i < \tau_{N_{T}-j+1}.
		\end{cases}
	\end{align}

	By construction, for each \(j=1,\ldots,N_{T}\), the two consecutive schedules
	\(\eta^{(j-1)}\) and \(\eta^{(j)}\) differ at exactly one time instant
	and they coincide at all other times.
	The difference of the sums of the \gls{dpsc}s admits the telescoping decomposition
	\begin{align}
		D_{T}^{\tilde{\pi}}(\lambda)-D_{T}^{\pi^{(1)}}(\lambda)
		=
		\sum_{j=1}^{N_{T}}
		\left(
		D_{T}^{\eta^{(j-1)}}(\lambda)-D_{T}^{\eta^{(j)}}(\lambda)
		\right).
		\label{eq:telescopingHybridDecomposition}
	\end{align}
	Next, we want to show that under the conditions of the theorem
	we have
	\begin{align}
		\limsup_{T \to \infty} \frac{1}{T} D_T^{\tilde{\pi}}
		\geq \limsup_{T \to \infty} \frac{1}{T} D_T^{\pi^{(1)}}.
	\end{align}
	A sufficient condition for this is that the liminf
	of the difference of the two sequences is non-negative.
	We write this difference as follows and identify
	some $o(1)$ terms:
	\begin{align}
		 & \liminf_{T\to\infty} \frac{D_{T}^{\tilde{\pi}}(\lambda)-D_{T}^{\pi^{(1)}}(\lambda)}{T} \\
		 & =
		\liminf_{T\to\infty} \frac{1}{T}\sum_{j=1}^{N_{T}}
		\left(
		D_{T}^{\eta^{(j-1)}}(\lambda)-D_{T}^{\eta^{(j)}}(\lambda)
		\right)                                                                                   \\
		 & =
		\liminf_{T\to\infty} \Bigg\{ \underbrace{\frac{1}{T}\!\!\sum_{j=1}^{N_T-N_{T_{\mathrm{c}}}}
			\!\!\left(
			D_{T}^{\eta^{(j-1)}}(\lambda)-D_{T}^{\eta^{(j)}}(\lambda)
		\right)}_{o(1)}                                                                           \\
		 & +
		\frac{1}{T}\sum_{j=N_T-N_{T_{\mathrm{c}}}+1}^{N_{T}}
		\left(
		D_{T}^{\eta^{(j-1)}}(\lambda)-D_{T}^{\eta^{(j)}}(\lambda)
		\right) \Bigg\}
		\label{eq:limsupTelescopingDecomposition}
	\end{align}

	Note that the $o(1)$ term vanishes as $T \to \infty$, and therefore
	we can restrict our attention to (\ref{eq:limsupTelescopingDecomposition}).
	What is left to show is that each summand in (\ref{eq:limsupTelescopingDecomposition})
	is strictly positive under the conditions of the theorem.
	If this holds, the average age of an arbitrary schedule is equal to
	or larger than the average age of the constant schedule that always chooses $n_1$.
	Next, fix $j\in\{N_T - N_{T_{\mathrm{c}}}+1,\ldots,N_{T}\}$ and consider the pair
	\(\eta^{(j-1)}\) and \(\eta^{(j)}\). To simplify the notation,
	we define
	\begin{align}
		\eta^+  \triangleq \eta^{(j)},\quad
		\eta^-  \triangleq \eta^{(j-1)},\quad
		t       \triangleq \tau_{N_{T}-j+1}.
	\end{align}
	Let \(m_i^+\) and \(m_i^-\) denote the expected
	age at the beginning of transmission $i$ under
	\(\eta^+\) and \(\eta^-\), respectively.
	Since the two schedules coincide before time \(t\), we write
	\(m_t\triangleq m_t^+=m_t^-\).
	Define
	\begin{align}
		 & f_1(m)  =n_1+\epsilon_1 m,
		\qquad
		f_2(m)=n_2+\epsilon_2 m,                                    \\
		 & \Delta_k  \triangleq m^-_{t+k}-m^+_{t+k}, \qquad k\ge 0,
	\end{align}
	so that \(\Delta_0=0\). The one-step difference in
	the expected age progression is
	\begin{align}\label{eq:UnifiedDelta1}
		\Delta_1
		= f_2(m_t)-f_1(m_t)
		= \delta_n-\delta_{\epsilon}\, m_t.
	\end{align}
	For \(k\ge 1\), both schedules use the same action $a_{t+k} = n_1$, so that we have the recursion
	\begin{align}\label{eq:UnifiedGapRecursion}
		\Delta_{k+1} = \epsilon_1 \Delta_k.
	\end{align}
	The \gls{dpsc} for an action \(e\in\{1,2\}\) and expected age \(m\) is given by
	\begin{align}
		d_e(m) \triangleq n_e m + \frac{n_e^2}{2} - \lambda n_e.
	\end{align}
	Therefore, the immediate difference at time \(t\) is
	\begin{align}\label{eq:UnifiedImmediateDifference}
		d_2(m_t)-d_1(m_t)
		=
		\delta_n\,m_t + \frac{n_2^2-n_1^2}{2} - \lambda \delta_n.
	\end{align}
	We consider two regimes. In the first regime, we know
	that, independent of the specific action sequence before time $t$, for time indices $t$ large enough,
	the action $n_1$ increases the next expected age, so that we have $m_{t+1}>m_t$.
	In the second regime, we have the opposite, so that action $n_1$ decreases the next expected age.

	\noindent\textbf{(i) Above-threshold regime.}
	Assume $m_L>m^*=\frac{\delta_n}{\delta_{\epsilon}}$.
	By the lower bound in Lemma \ref{lem:boundedcostLemmaAoI}, there exists a finite time \(T^*\) such that along every trajectory,
	\begin{align}
		m_i>m^* \qquad \text{for all } i>T^*.
	\end{align}
	A finite prefix does not affect the average-cost
	criterion, so we may restrict attention to times  \(t\ge T^*\).
	For \(m_t>m^*\), \eqref{eq:UnifiedDelta1} gives
	\(\Delta_1<0\), i.e., choosing \(n_2\) produces a smaller
	next-state expected age than choosing \(n_1\).
	From \eqref{eq:UnifiedGapRecursion} we get
	\begin{align}
		\Delta_k = \epsilon_1^{k-1}\Delta_1,\qquad k\ge 1.
	\end{align}
	For \(k\ge 1\), since both schedules use the same action,
	\begin{align}
		\left|d_{a_{t+k}}(m^+_{t+k})-d_{a_{t+k}}(m^-_{t+k})\right|
		= n_1 |\Delta_k|.
	\end{align}
	Hence, the future advantage of \(\eta^-\) can be upper-bonded as
	\begin{align}
		\sum_{k\ge 1}\left(d_{a_{t+k}}(m^+_{t+k})\!-d_{a_{t+k}}(m^-_{t+k})\right)
		\leq \frac{n_1}{\bar{\epsilon}_1}\bigl(f_1(m_t)\!-f_2(m_t)\bigr).
	\end{align}
	Therefore, the total accumulated cost increase caused by choosing \(n_2\) instead of \(n_1\) at time \(t\) is bounded from below by
	\begin{align}
		\Psi_\lambda(m_t)
		 & \triangleq
		\bigl(d_2(m_t)-d_1(m_t)\bigr)
		-
		\frac{n_1}{\bar{\epsilon}_1}\bigl(f_1(m_t)-f_2(m_t)\bigr) \\
		 & =
		\delta_n\,m_t \!+\! \frac{n_2^2-n_1^2}{2} - \lambda \delta_n
		-\frac{n_1}{\bar{\epsilon}_1}\bigl(-\delta_n+\delta_{\epsilon}\,m_t\bigr).
		\label{eq:Psi_n1_Optimality}
	\end{align}
	Since \(\Psi_\lambda\) is decreasing in \(\lambda\), it is sufficient to require
	$\Psi_U(m_t)>0$,
	where \(U\) is the upper bound on \(\lambda^*\) in Lemma \ref{lem:AoIBounds}.
	Moreover, in this regime \(m_t\in\left(m^*,\,m_U\right)\), and \(\Psi_U(m)\) is affine in \(m\) with negative slope
	\begin{align}
		\delta_n-\frac{n_1\delta_{\epsilon}}{\bar{\epsilon}_1} = \delta_{\epsilon}\left( m^*-\frac{n_1}{\bar{\epsilon}_1}\right)<0.
	\end{align}
	Hence, its minimum over the reachable interval is attained at \(m=m_U\). Evaluating there yields the sufficient condition
	\begin{align}
		U<\frac{n_1+n_2}{2}+\frac{n_1}{\bar{\epsilon}_1} + m_U\left( 1-\frac{n_1\delta_{\epsilon}}{\bar{\epsilon}_1\delta_n}\right).
	\end{align}
	Under this condition, every occurrence of \(n_2\) after time \(T^*\) can be replaced by
	\(n_1\) with a strictly positive reduction in accumulated \gls{dpsc}.
	\medskip

	\noindent\textbf{(ii) Below-threshold regime.}
	Assume
	$m_U<m^*=\frac{\delta_n}{\delta_{\epsilon}}$.
	Then, similarly to the above-threshold regime, for
	all $i > T^*$ with $T^*$ large enough, we have
	$m_i \le m_U<m^*$,
	so \eqref{eq:UnifiedDelta1} implies \(\Delta_1>0\).
	Thus, choosing \(n_2\) produces a larger next-state
	expected age than choosing \(n_1\).
	Following the same steps as in the above-threshold regime,
	we get the following lower bound for the total cost increase
	caused by choosing $n_2$ instead of $n_1$ at time $t$:
	\begin{align}
		\Psi_\lambda(m_t) & =
		\delta_n m_t + \frac{n_2^2-n_1^2}{2} - \lambda \delta_n                      \\
		                  & -n_1\left( \frac{1}{\bar{\epsilon}_1} - \delta_M \right)
		(-\delta_n + \delta_{\epsilon} m_t),
	\end{align}
	with $\delta_M = \frac{\epsilon_1^M}{\bar{\epsilon}_1}$. Note that
	$\delta_M$ can be made arbitrarily small by choosing the tail length
	$M$ sufficiently large. Therefore, we can omit the term in the following.
	Again \(\Psi_\lambda\) is decreasing in \(\lambda\),
	so it is sufficient to require
	$\Psi_U(m_t)>0$.
	The achievable expected ages satisfy
	\begin{align}
		m_t \in \left[m_L-\delta,\;m_U\right]
	\end{align}
	for an arbitrarily small \(\delta>0\), for $t$ sufficiently large.
	Since \(\Psi_U(m)\) is affine, its minimum is attained at one of the endpoints. Its slope is
	\begin{align}
		s_\Psi\triangleq\delta_n-\frac{n_1\delta_{\epsilon}}{\bar{\epsilon}_1} = \delta_{\epsilon}\left( m^*-\frac{n_1}{\bar{\epsilon}_1}\right).
	\end{align}
	Under the conditions of the theorem, we have \(s_\Psi\ge 0\), and therefore
	the minimum is attained at \(m=m_L\). Letting \(\delta\downarrow 0\) yields the sufficient condition
	\begin{align}
		U<\frac{n_1+n_2}{2}+\frac{n_1}{\bar{\epsilon}_1} + m_L\left( 1-\frac{n_1\delta_{\epsilon}}{\bar{\epsilon}_1\delta_n}\right).
	\end{align}
	Under these conditions, each occurrence of \(n_2\) can be replaced by \(n_1\) with a strictly positive reduction in accumulated Dinkelbach cost. Hence, always choosing \(n_1\) is optimal.
\end{proof}

\section{Proof of Theorem \ref{thm:chooseN2OptimalUnified}}\label{app:chooseN2OptimalUnified}
\begin{proof}
	We begin with the same setup as in the proof of Theorem \ref{thm:chooseN1OptimalUnified},
	but the schedules $\eta^{(j)}$ now replace the actions of an arbitrary
	different schedule by \(n_2\) instead of \(n_1\). For the local comparison,
	let \(\mu^-\) choose \(n_1\) at time \(t\) and let \(\mu^+\) choose \(n_2\)
	at time \(t\), with both schedules identical afterward and before. Let \(m_i^-\) and
	\(m_i^+\) denote the expected age at the beginning of transmission $i$ under \(\mu^-\) and \(\mu^+\),
	respectively, and write \(m_t\triangleq m_t^-=m_t^+\).
	Let $f_1$ and $f_2$ be defined as in the proof of Theorem~\ref{thm:chooseN1OptimalUnified}.
	We write
	\begin{align}
		\Delta_k \triangleq m^-_{t+k}-m^+_{t+k}, \qquad k\ge 0,
	\end{align}
	so that \(\Delta_0=0\). The one-step difference is
	\begin{align}\label{eq:ChooseN2Delta1}
		\Delta_1
		= f_1(m_t)-f_2(m_t)
		= -\delta_n+\delta_{\epsilon}\, m_t.
	\end{align}
	Assume $m_L>m^*=\frac{\delta_n}{\delta_{\epsilon}}$.
	Then, by Lemma \ref{lem:boundedcostLemmaAoI}, there exists a finite time \(T^*\) such that along every trajectory,
	\begin{align}
		m_i>m^* \qquad \text{for all } i>T^*.
	\end{align}
	A finite prefix does not affect the average-cost criterion, so
	we may restrict attention to times \(t\ge T^*\).
	In particular, \eqref{eq:ChooseN2Delta1} implies that $\Delta_1>0$.
	For \(k\ge 1\), both schedules use the same action \(a_{t+k} = n_2\), hence
	$\Delta_{k+1}=\epsilon_2\Delta_k$.
	%Since \(\epsilon(a_{t+k})\ge \epsilon_2\), we obtain
	%\begin{align}
	%		\Delta_k \ge \epsilon_2^{k-1}\Delta_1,\qquad k\ge 1.
	%\end{align}
	Let the \gls{dpsc} be denoted by
	\begin{align}
		d_e(m)\triangleq n_e m+\frac{n_e^2}{2}-\lambda n_e,\qquad e\in\{1,2\}.
	\end{align}
	The immediate difference at time \(t\) is
	\begin{align}
		d_1(m_t)-d_2(m_t)
		=
		-\delta_n\,m_t-\frac{n_2^2-n_1^2}{2}+\lambda\delta_n.
	\end{align}
	For \(k\ge 1\), since both schedules use the same action,
	\begin{align}
		d_{a_{t+k}}(m^-_{t+k})-d_{a_{t+k}}(m^+_{t+k})
		= n_2\Delta_k.
	\end{align}
	Therefore, the total accumulated cost decrease caused by choosing \(n_2\) instead of \(n_1\) at time \(t\) is bounded from below by
	\begin{align}
		\Psi_\lambda(m_t)
		 & \triangleq
		\bigl(d_1(m_t)-d_2(m_t)\bigr)
		+
		\left(\frac{1}{\bar{\epsilon}_2}-\delta_M \right)n_2\Delta_1                                         \\
		 & =
		-\delta_n\,m_t-\frac{n_2^2-n_1^2}{2}+\lambda\delta_n                                                 \\
		 & +\left(\frac{1}{\bar{\epsilon}_2}-\delta_M\right)n_2\bigl(-\delta_n+\delta_{\epsilon}\,m_t\bigr),
	\end{align}
	where $\delta_M = \frac{\epsilon_2^M}{\bar{\epsilon}_2}$ can be made arbitrarily
	small by choosing $M$ sufficiently large.
	Since \(\Psi_\lambda\) is increasing in \(\lambda\), it is sufficient to require
	$\Psi_L(m_t)>0$.
	Moreover, for all sufficiently large \(t\), we have
	\begin{align}
		m_t \in \left[m_L-\delta,\;m_U \right],
	\end{align}
	for an arbitrarily small \(\delta>0\).
	The function \(\Psi_L(m)\) is affine with slope
	$\frac{n_2\delta_{\epsilon}}{\bar{\epsilon}_2}-\delta_n$.
	Under the assumption of the theorem, we have
	\begin{align}
		\frac{n_2}{\bar{\epsilon}_2}> m_L > \frac{\delta_n}{\delta_{\epsilon}} = m^*,
	\end{align}
	and therefore this slope is positive.
	Hence, the minimum over the reachable interval
	is attained at the left endpoint.
	Letting \(\delta\downarrow 0\),
	it is therefore sufficient that
	$\Psi_L\!\left(m_L\right)>0$.
	After simplification, this condition is equivalent to
	\begin{align}
		L> \frac{n_1+n_2}{2} + \frac{n_2}{\bar{\epsilon}_2} + m_L \left( 1-\frac{n_2\delta_{\epsilon}}{\bar{\epsilon}_2\delta_n}\right).
	\end{align}
	Under these conditions, each occurrence of
	\(n_1\) can be replaced by \(n_2\) with a strictly positive reduction in accumulated \gls{dpsc}.
	Hence, always choosing \(n_2\) is optimal.

\end{proof}

\section{Proof of Theorem \ref{thm:chooseN1OptimalExpUnified}}\label{app:chooseN1OptimalExpUnified}
\begin{proof}
	The proof is quite similar to the
	proof of Theorem \ref{thm:chooseN1OptimalUnified}.
	Therefore, we first state analogues of
	Lemma \ref{lem:boundedcostLemmaAoI}
	and Lemma \ref{lem:AoIBounds}
	for the expected exponential age.

	\begin{lemma}\label{lem:BoundsOnZi}
		The
		expected exponential age
		\begin{align}
			z_i \triangleq \sum_{x\in\mathcal S} b_i(x)e^{\alpha x}
		\end{align}
		satisfies
		$z_L -\delta
			\le z_i \le z_U$,
		where $\delta > 0$ is arbitrary and $i$ is sufficiently large, and
		\begin{align}
			z_L = \min \left\{ \frac{\beta_1}{\bar{\rho}_1}
			, \frac{\beta_2}{\bar{\rho}_2}\right\}, \quad
			z_U =\max \left\{ \frac{\beta_1}{\bar{\rho}_1}
			, \frac{\beta_2}{\bar{\rho}_2}\right\}
		\end{align}
	\end{lemma}

	\begin{proof}

		We start with the expression for
		the expected exponential age $z_i$,
		which is given as
		\begin{align}
			z_i = \sum_{x\in\mathcal S} b_i(x)e^{\alpha x}
			= \mathbb E[e^{\alpha X_i}],
		\end{align}
		where \(X_i\) denotes the age at the beginning of the \(i\)-th transmission.
		Next, we derive a recursion for \(z_i\). Conditioned on action \(a_i\),
		\begin{align}
			X_{i+1} =
			\begin{cases}
				a_i,       & \text{with probability } \bar{\epsilon}(a_i), \\[1mm]
				X_i + a_i, & \text{with probability } \epsilon(a_i).
			\end{cases}
		\end{align}
		Therefore,
		\begin{align}
			e^{\alpha X_{i+1}} =
			\begin{cases}
				e^{\alpha a_i}, & \text{with probability } \bar{\epsilon}(a_i), \\
				e^{\alpha a_i}e^{\alpha X_i},
				                & \text{with probability } \epsilon(a_i).
			\end{cases}
		\end{align}
		Taking expectations yields
		\begin{align}\label{eq:RecursionForZi}
			z_{i+1}
			 & = \bar{\epsilon}(a_i)e^{\alpha a_i}
			+ \epsilon(a_i)e^{\alpha a_i} z_i      \\
			 & = \beta(a_i) + \rho(a_i) z_i,
		\end{align}
		with
		$\beta(a_i)\triangleq\bar{\epsilon}(a_i)e^{\alpha a_i}$,
		and
		$\rho(a_i)\triangleq\epsilon(a_i)e^{\alpha a_i}.$
		Now, we define the two maps $h_e(z) \triangleq \beta_e + \rho_e z$ for $e\in\{1,2\}$. Note that these
		maps have unique fixed points at $\frac{\beta_e}{\bar{\rho}_e}$
		for $z \geq 0$ if $\rho_e < 1$.
		By the same
		arguments as in the proof of Lemma~\ref{lem:boundedcostLemmaAoI},
		the theorem follows.
	\end{proof}

	Next, we derive lower and upper bounds on the average exponential age:
	\begin{lemma}\label{lem:ExponentialAoIBounds}
		Let $\Delta_{\exp,\mathrm{avg}}^*$ be the minimal average exponential \gls{aoi}
		achievable by any schedule $\pi \in \Pi_D$. Then the following bounds hold:
		\begin{align}
			 & \max\Bigg\{
			\frac{\bar{\epsilon}_2e^{\alpha n_1}
				\left(e^{\alpha n_1}-1\right)}
			{\alpha n_1\left(1-\epsilon_2 e^{\alpha n_1}\right)}, \nonumber \\
			 & \min_{\substack{\theta_1,\theta_2\ge 0                       \\
					\frac{n_1}{\bar{\epsilon}_1}\theta_1+
					\frac{n_2}{\bar{\epsilon}_2}\theta_2\le 1}}
			\Bigg\{
			\frac{e^{\alpha n_1}}{\alpha}
			(\theta_1+\theta_2)
			\left(
			e^{\frac{\alpha}{\theta_1+\theta_2}}-1
			\right) \nonumber                                               \\
			 &
			+
			\frac{e^{\alpha n_2}-e^{\alpha n_1}}{\alpha}
			\left(e^{\alpha n_1}-1\right)\theta_2
			\Bigg\}
			\Bigg\}                                                         \\
			 & \le \Delta_{\exp,\mathrm{avg}}^* \nonumber
			\le
			\min_{i\in \{1,2\}}
			\frac{\bar{\epsilon}_ie^{\alpha n_i}
				\left(e^{\alpha n_i}-1\right)}
			{\alpha n_i\left(1-\epsilon_i e^{\alpha n_i}\right)}
			\label{eq:expPenaltyBounds}
		\end{align}
	\end{lemma}

	\begin{proof}
		The upper bound is obtained by taking the minimum of the two constant schedules that always
		choose \(n_1\) or always choose \(n_2\).
		The lower bound is the combination of two bounds.
		The first bound is the average exponential
		\gls{aoi} of a blocklength error probability tuple that is given by
		the shortest blocklength and lowest error probability available in $\mathcal{C}$.
		We now prove the second lower bound.
		Let $Y_0,Y_1,Y_2,\ldots$, the time horizon
		$T'$ and the actions $a_{i-}$ be defined
		as in the proof of Lemma~\ref{lem:AoIBounds}.
		The area $Q_i^{\exp}$ under the exponential \gls{aoi} curve
		during the interval $Y_i$ is
		\begin{align}
			Q_i^{\exp}
			 & =
			\int_0^{Y_i} e^{\alpha(a_{i-}+t)}\,dt \\
			 & =
			\frac{e^{\alpha a_{i-}}}{\alpha}
			\left(e^{\alpha Y_i}-1\right).\label{eq:QExponential}
		\end{align}
		Therefore, the expected average exponential
		\gls{aoi} under
		any admissible schedule can be lower-bounded as
		\begin{align}
			\Delta_{\exp,\mathrm{avg}}(\pi)
			 & \geq
			\liminf_{T'\to\infty} \frac{1}{T'}
			\mathbb{E}\Bigg[
				\sum_{i=0}^{K_{T'}-1}
				\frac{e^{\alpha a_{i-}}}{\alpha}
			\left(e^{\alpha Y_i}-1\right) \\
			 & +
				\frac{e^{\alpha a_{K_{T'}}}}{\alpha}
				\left(e^{\alpha R_{T'}}-1\right)
				\Bigg]
			\coloneq \liminf_{T'\to\infty} Z_{T'},
		\end{align}
		where $R_{T'}$ denotes the remaining time interval of $T'$
		after the last successful transmission,
		$K_{T'}$ denotes the number of successful transmissions
		and $a_{K_{T'}}$ denotes the used blocklength at
		the beginning of the remainder interval $R_{T'}$.
		Since $a_{i-}\in\{n_1,n_2\}$, we can write
		\begin{align}
			e^{\alpha a_{i-}}
			 & =
			e^{\alpha n_1}
			+
			\left(e^{\alpha n_2}-e^{\alpha n_1}\right)
			\mathbf{1}_{\{a_{i-}=n_2\}} .
		\end{align}
		Therefore,
		\begin{align}
			Z_{T'}
			 & \ge
			\frac{e^{\alpha n_1}}{\alpha T'}
			\mathbb{E}\left[
				\sum_{i=0}^{K_{T'}-1}
				\left(e^{\alpha Y_i}-1\right)
				+
				\left(e^{\alpha R_{T'}}-1\right)
				\right]
			\nonumber \\
			 & \quad+
			\frac{e^{\alpha n_2}-e^{\alpha n_1}}{\alpha T'}
			\mathbb{E}\left[
				\sum_{i=0}^{K_{T'}-1}
				\mathbf{1}_{\{a_{i-}=n_2\}}
				\left(e^{\alpha Y_i}-1\right)
				\right].
			\label{eq:expDecompositionFiniteHorizon}
		\end{align}

		If the average exponential \gls{aoi}
		is finite, we can choose a discrete
		sequence \(T'_m\to\infty\), so that
		the \(\liminf\) of \(Z_{T'}\)
		is attained:
		\begin{align}
			Z_{T'_m}\to \liminf_{T'\to\infty} Z_{T'} .
		\end{align}
		Now, we apply the Bolzano-Weierstrass theorem to extract a further subsequence for which the
		expected success rates, denoted by $\theta_i$, converge:
		\begin{align}
			\lim_{T'_m \to\infty} \mathbb{E}\left[\frac{S_i(T'_m)}{T'_m}\right] = \theta_i, \qquad i\in\{1,2\},
		\end{align}
		where $S_i(T'_m)$ denotes the number of successful transmissions
		with blocklength $n_i$ up to time $T'_m$.
		This is valid because the expected success rates are
		bounded in $[0,1/n_1]$.
		Note that we have not relabeled the subsequence to ease the notation.
		Now, we take a look at the terms in
		\eqref{eq:expDecompositionFiniteHorizon} separately.
		For the first term, Jensen's inequality applied to the sample average gives
		\begin{align}
			\frac{1}{K_{T'_m}+1}
			\left(
			\sum_{i=0}^{K_{T'_m}-1} e^{\alpha Y_i}
			+
			e^{\alpha R_{T'_m}}
			\right)
			\ge
			e^{\alpha T'_m/(K_{T'_m}+1)} .
		\end{align}
		Therefore,
		\begin{align}
			 & \frac{1}{T'_m}
			\mathbb{E}\left[
				\sum_{i=0}^{K_{T'_m}-1}
				\left(e^{\alpha Y_i}-1\right)
				+
				\left(e^{\alpha R_{T'_m}}-1\right)
				\right]
			\nonumber         \\
			 & \ge
			\mathbb{E}\left[
				\frac{K_{T'_m}+1}{T'_m}
				\left(
				e^{\alpha T'_m/(K_{T'_m}+1)}-1
				\right)
				\right].
		\end{align}
		Define $\phi(x)=x\left(e^{\alpha/x}-1\right),\quad x>0$.
		Since $\phi''(x)
			=
			\frac{\alpha^2}{x^3}e^{\alpha/x}
			>0$,
		the function $\phi$ is convex. Hence, by Jensen's inequality again,
		\begin{align}
			 & \mathbb{E}\left[
				\frac{K_{T'_m}+1}{T'_m}
				\left(
				e^{\alpha T'_m/(K_{T'_m}+1)}-1
				\right)
				\right]
			\nonumber           \\
			 & \ge
			\mathbb{E}\left[\frac{K_{T'_m}+1}{T'_m}\right]
			\left(
			e^{
				\alpha/
				\mathbb{E}\left[(K_{T'_m}+1)/T'_m\right]}
			-1
			\right).
		\end{align}
		Since $K_{T'_m}=S_1(T'_m)+S_2(T'_m)$, we obtain
		\begin{align}
			\mathbb{E}\left[\frac{K_{T'_m}+1}{T'_m}\right]
			\to
			\theta_1+\theta_2 .
		\end{align}
		Consequently,
		\begin{align}
			 & \liminf_{m\to\infty}
			\frac{e^{\alpha n_1}}{\alpha T'_m}
			\mathbb{E}\left[
				\sum_{i=0}^{K_{T'_m}-1}
				\left(e^{\alpha Y_i}-1\right)
				+
				\left(e^{\alpha R_{T'_m}}-1\right)
				\right]
			\nonumber               \\
			 & \ge
			\frac{e^{\alpha n_1}}{\alpha}
			(\theta_1+\theta_2)
			\left(
			e^{\frac{\alpha}{\theta_1+\theta_2}}-1
			\right),
		\end{align}

		For the second term in
		\eqref{eq:expDecompositionFiniteHorizon}, every inter-success interval has
		length at least $n_1$ and hence,
		$e^{\alpha Y_i}-1
			\ge
			e^{\alpha n_1}-1$.
		Therefore,
		\begin{align}
			 & \liminf_{m\to\infty}
			\frac{1}{T'_m}
			\mathbb{E}\left[
				\sum_{i=0}^{K_{T'_m}-1}
				\mathbf{1}_{\{a_{i-}=n_2\}}
				\left(e^{\alpha Y_i}-1\right)
				\right]
			\nonumber               \\
			 & \ge
			\liminf_{m\to\infty}
			\left(e^{\alpha n_1}-1\right)
			\frac{1}{T'_m}
			\mathbb{E}\left[
				\sum_{i=0}^{K_{T'_m}-1}
				\mathbf{1}_{\{a_{i-}=n_2\}}
				\right]
			\nonumber               \\
			 & =
			\liminf_{m\to\infty}
			\left(e^{\alpha n_1}-1\right)
			\mathbb{E}\left[\frac{S_2(T'_m)}{T'_m}\right]
			\nonumber               \\
			 & =
			\left(e^{\alpha n_1}-1\right)\theta_2 .
		\end{align}

		Combining the two terms gives, for every admissible schedule,
		\begin{align}
			\Delta_{\exp,\mathrm{avg}}(\pi)
			 & \ge
			\frac{e^{\alpha n_1}}{\alpha}
			(\theta_1+\theta_2)
			\left(
			e^{\frac{\alpha}{\theta_1+\theta_2}}-1
			\right)
			\nonumber \\
			 & \quad+
			\frac{e^{\alpha n_2}-e^{\alpha n_1}}{\alpha}
			\left(e^{\alpha n_1}-1\right)\theta_2 .
			\label{eq:expLowerForPairRewritten}
		\end{align}
		Optimizing over the same set of
		admissible rate pairs as in Lemma
		\ref{lem:AoIBounds} finishes the bound.
	\end{proof}

	Now, we use the same replacement argument as in the proof of Theorem
	\ref{thm:chooseN1OptimalUnified}. For the local comparison, let \(\mu^+\)
	choose \(n_1\) at time \(t\) and let \(\mu^-\) choose \(n_2\) at time
	\(t\), with both schedules identical before and after time \(t\).
	Let \(z_i^+\) and \(z_i^-\) denote the corresponding
	expected exponential ages, and write \(z_t\triangleq z_t^+=z_t^-\).
	Define the two maps
	\begin{align}
		h_1(z)=\beta_1+\rho_1 z,
		\qquad
		h_2(z)=\beta_2+\rho_2 z.
	\end{align}
	Using (\ref{eq:QExponential}) and the definition
	of the expected exponential ages in Lemma \ref{lem:BoundsOnZi},
	the \gls{dpsc}
	at state \(z_i\) is given by
	\begin{align}
		d_e(z_i)=K_e z_i-\lambda n_e, \qquad e\in\{1,2\}.
	\end{align}
	Let
	$\Delta_k \triangleq z^-_{t+k}-z^+_{t+k}$, with
	$k\ge 0$,
	so that $\Delta_0=0$. The one-step difference is
	\begin{align}\label{eq:ExpUnifiedDelta1}
		\Delta_1
		= h_2(z_t)-h_1(z_t)
		= (\beta_2-\beta_1)-(\rho_1-\rho_2)z_t.
	\end{align}
	For $k\ge 1$, both schedules use the same future actions.
	Hence, we have
	\begin{align}\label{eq:ExpUnifiedGapRecursion}
		\Delta_{k+1}=\rho(a_{t+k})\Delta_k.
	\end{align}

	At time $t$, the immediate difference in the \gls{dpsc} is
	\begin{align}\label{eq:ExpUnifiedImmediateDiff}
		d_2(z_t)-d_1(z_t)
		=
		(K_2-K_1)z_t-\lambda\delta_n.
	\end{align}
	We now prove the statement for the case $\rho_2 < \rho_1 < 1$
	and analyze the same two regimes as
	in the proof of
	Theorem~\ref{thm:chooseN1OptimalUnified}.
	The case $\rho_1 < \rho_2 < 1$ is treated
	at the end of the proof.

	\noindent\textbf{(i) Above-threshold regime.}
	Assume $z_L>z^*$.
	Since $h_2(z)-h_1(z)=(\beta_2-\beta_1)-(\rho_1-\rho_2)z$ and $\rho_1>\rho_2$, the two maps cross exactly once at $z=z^*$.
	By Lemma \ref{lem:BoundsOnZi}, the lower bound on $z_i$ converges to $z_L$, so there exists a finite time $T^*$ such that
	\begin{align}
		z_i>z^* \qquad \text{for all } i>T^*.
	\end{align}
	A finite prefix does not affect the average-cost criterion, so we may
	restrict our attention to $t\ge T^*$.
	For $z_t>z^*$, \eqref{eq:ExpUnifiedDelta1} implies
	\(\Delta_1<0\), i.e., choosing $n_2$
	produces a smaller expected exponential age than choosing $n_1$.
	Because both schedules use the same action $n_1$ for $k\ge 1$,
	from \eqref{eq:ExpUnifiedGapRecursion} we get
	\begin{align}
		 & |\Delta_k| = \rho_1^{k-1}|\Delta_1|,\qquad k\ge 1,         \\
		 & \left|d_{a_{t+k}}(z^+_{t+k})-d_{a_{t+k}}(z^-_{t+k})\right|
		= K_1 |\Delta_k|.
	\end{align}
	Therefore, the future advantage of schedule \(\mu^-\) can be upper bounded as
	\begin{align}
		 & \sum_{k\ge 1}\Bigl(d_{a_{t+k}}(z^+_{t+k})-d_{a_{t+k}}(z^-_{t+k})\Bigr) \\
		 & \leq
		\frac{K_1}{\bar{\rho}_1}\bigl(h_1(z_t)-h_2(z_t)\bigr).
	\end{align}
	Hence, the total accumulated cost increase caused by choosing $n_2$ instead of $n_1$ at time $t$
	is bounded from below by
	\begin{align}
		 & \Psi_\lambda(z_t)
		\triangleq
		\bigl(d_2(z_t)-d_1(z_t)\bigr)
		-
		\frac{K_1}{\bar{\rho}_1}\bigl(h_1(z_t)-h_2(z_t)\bigr) \\
		 & =
		(K_2-K_1)z_t-\lambda\delta_n                          \\
		 & -
		\frac{K_1}{\bar{\rho}_1}
		\Bigl((\beta_1-\beta_2)+(\rho_1-\rho_2)z_t\Bigr).
	\end{align}
	Since $\Psi_\lambda$ is decreasing in $\lambda$,
	it is sufficient to require
	$\Psi_{U_{\exp}}(z_t)>0$.
	Moreover, in this regime we have
	\begin{align}
		z_t \in [z_L-\delta,\,z_U]
	\end{align}
	for an arbitrarily small $\delta>0$ and all
	sufficiently large $t$. Since $\Psi_{U_{\exp}}(z)$ is affine in $z$,
	its minimum over the achievable interval is attained at the right or left endpoint. Letting $\delta\downarrow 0$, a sufficient condition is
	\begin{align}
		\min_{z \in [z_L, z_U]} \Psi_{U_{\exp}}(z)>0.
	\end{align}
	This is exactly
	\begin{align}
		U_{\exp}
		< &
		\frac{1}{\delta_n} \min_{z \in [z_L, z_U]} \Bigg\{ (K_2-K_1)z                     \\
		  & - \frac{K_1}{\bar{\rho}_1}\Big(\beta_1-\beta_2+(\rho_1-\rho_2)z\Big) \Bigg\}.
	\end{align}
	Under this condition, every occurrence of
	$n_2$ after time $T^*$ can be replaced by
	$n_1$ with a strictly positive reduction in
	the accumulated \gls{dpsc}. Hence, always choosing $n_1$ is optimal.

	\medskip
	\noindent\textbf{(ii) Below-threshold regime.}

	Assume $z_U <z^*$.
	Then Lemma \ref{lem:BoundsOnZi} implies $z_i <z^*$
	for all reachable states, and therefore \eqref{eq:ExpUnifiedDelta1} gives $\Delta_1>0$.
	Following the same lines
	as in the above-threshold regime, we
	get the lower bound
	\begin{align}
		 & \Psi_{U_{\exp}}(z_t)
		=
		(K_2-K_1)z_t-U_{\exp}\delta_n \\
		 & +
		K_1 \left( \frac{1}{\bar{\rho}_1} - \delta_M \right)
		\Bigl((\beta_2-\beta_1)-(\rho_1-\rho_2)z_t\Bigr),\label{eq:ExpAoITotalCostRegime2}
	\end{align}
	with $\delta_M = \frac{\rho_1^M}{\bar{\rho}_1}$, which can be made
	arbitrarily small by choosing $M$ sufficiently large. Continuing with the bounding
	as in the above-threshold regime gives the second part in the theorem.

	\textbf{(iii) Case $\rho_1 < \rho_2 < 1$.}

	For the case $\rho_1 < \rho_2 < 1$, it remains to show
	that (\ref{eq:ExpAoITotalCostRegime2}) is strictly positive
	in this regime. The term $\delta_M$ again can be made arbitrarily small for
	$M$ sufficiently large and can be omitted in the following.
	Note that the upper bound $U_{\exp}$ is given by
	$U_{\exp}
		=
		\frac{K_1\beta_1}{n_1\bar{\rho}_1}$,
	and \eqref{eq:ExpAoITotalCostRegime2} can be rewritten as
	\begin{align}
		\Psi_{U_{\exp}}(z_t)
		 & =
		\left(
		K_2-K_1+\frac{K_1(\rho_2-\rho_1)}{\bar{\rho}_1}
		\right)z_t \\
		 & -
		U_{\exp}\delta_n
		+
		\frac{K_1(\beta_2-\beta_1)}{\bar{\rho}_1}.
	\end{align}
	Since the second summand in the first bracket is positive, this can be bounded from below by
	\begin{align}
		\Psi_{U_{\exp}}(z_t)
		 & \ge
		(K_2-K_1)z_t
		-
		\frac{K_1\beta_1}{n_1\bar{\rho}_1}\delta_n .
	\end{align}
	Using the lower bound $z_t\ge z_L -\delta =  \frac{\beta_1}{\bar{\rho}_1} - \delta$ with
	an arbitrarily small $\delta>0$ for all sufficiently
	large $t$, we get
	\begin{align}
		 & \Psi_{U_{\exp}}(z_t)
		\ge
		(K_2-K_1)\!\left(\frac{\beta_1}{\bar{\rho}_1}\!-\!\delta\right)
		\!-\!
		\frac{K_1\beta_1}{n_1\bar{\rho}_1}\delta_n \\
		 & =
		\frac{\beta_1}{\bar{\rho}_1}
		\left(
		K_2-K_1-\frac{K_1}{n_1}\delta_n
		\right)
		-\delta(K_2-K_1)                           \\
		 & =
		\frac{\beta_1}{\bar{\rho}_1}
		\left(
		K_2-\frac{n_2}{n_1}K_1
		\right)
		-\delta(K_2-K_1).
	\end{align}
	Since \(n_2>n_1\) implies \(K_2/n_2>K_1/n_1\), the first term is strictly
	positive. Choosing \(\delta>0\) sufficiently
	small yields
	$\Psi_{U_{\exp}}(z_t)>0$.
\end{proof}

\section{Proof of Theorem \ref{thm:chooseN2OptimalExpUnified}}\label{app:chooseN2OptimalExpUnified}
\begin{proof}
	We use the replacement argument from the proof of Theorem
	\ref{thm:chooseN1OptimalUnified}. For the local comparison, let \(\mu^-\)
	choose \(n_1\) at time \(t\) and let \(\mu^+\) choose \(n_2\) at time
	\(t\), with both schedules identical before and after time \(t\).
	Let \(z_i^-\) and \(z_i^+\) denote the corresponding
	expected exponential ages, and write \(z_t\triangleq z_t^-=z_t^+\).
	Let the two maps $h_1(z)$, $h_2(z)$ and the per-stage
	Dinkelbach cost $d_e(z)$ be defined as
	in the proof of Theorem \ref{thm:chooseN1OptimalExpUnified}.
	We now define
	\begin{align} \Delta_k\triangleq z^-_{t+k}-z^+_{t+k},\qquad k\ge 0,
	\end{align} so that $\Delta_0=0$ and
	\begin{align} \Delta_1 = h_1(z_t)-h_2(z_t)
		= (\beta_1-\beta_2)+(\rho_1-\rho_2)z_t.
	\end{align}
	By the assumption $z_L>z^*$ and the lower bound
	from Lemma~\ref{lem:BoundsOnZi}, there exists a
	finite time $T^*$ such that
	\begin{align} z_i>z^* \qquad \text{for all } i>T^*.
	\end{align}
	Hence, for $t\ge T^*$, we have $\Delta_1>0$.
	As usual, a finite prefix does
	not affect average-cost optimality, so it is enough to consider
	such times.
	For $k\ge 1$, both schedules use the same future actions.
	Hence, we have $\Delta_{k+1}=\rho(a_{t+k})\Delta_k$.
	Since $\rho(a_{t+k}) = \rho_2$, we obtain
	\begin{align} \Delta_k = \rho_2^{k-1}\Delta_1,\qquad k\ge 1.
	\end{align}
	At time $t$, the immediate cost difference is
	\begin{align} d_1(z_t)-d_2(z_t)=\lambda\delta_n-(K_2-K_1)z_t.
	\end{align}
	For $k\ge 1$, the terms $-\lambda a_{t+k}$ cancel, and
	\begin{align}
		d_{a_{t+k}}(z^-_{t+k})-d_{a_{t+k}}(z^+_{t+k}) = K_2\Delta_k.
	\end{align}
	Therefore, the total accumulated cost decrease caused by
	choosing $n_2$ instead of $n_1$ at time $t$
	is lower-bounded by
	\begin{align}
		 & \Phi_\lambda(z_t) \triangleq \bigl(d_1(z_t)-d_2(z_t)\bigr)
		+ K_2 \left(\frac{1}{\bar{\rho}_2} - \delta_M\right) \Delta_1 \\
		 & = \lambda\delta_n-(K_2-K_1)z_t                             \\
		 & +  K_2 \left(\frac{1}{\bar{\rho}_2} - \delta_M\right)
		\Bigl((\beta_1-\beta_2)+(\rho_1-\rho_2)z_t\Bigr),
	\end{align}
	where $\delta_M = \frac{\rho_2^M}{\bar{\rho}_2}$ can be made arbitrarily
	small by choosing $M$ sufficiently large.
	Since $\Phi_\lambda$ is increasing in $\lambda$,
	it is sufficient to require
	\begin{align}
		\Psi_L(z_t)\triangleq\Phi_{L_{\exp}}(z_t)>0.
	\end{align}
	Moreover, for every $\delta>0$ and all sufficiently large $t$,
	\begin{align}
		z_t\in [z_L - \delta,\,z_U].
	\end{align}
	The function $\Psi_L(z)$ is affine in $z$, so its
	minimum over the reachable interval is
	attained at one of the endpoints. Therefore, the sufficient
	condition becomes
	\begin{align}
		L_{\exp} > & \frac{1}{\delta_n} \max_{z \in [z_L, z_U]}
		\Bigg\{\frac{K_2(\beta_2-\beta_1)}{\bar{\rho}_2}        \\
		           & + z\Bigg( (K_2-K_1)
		-\frac{K_2(\rho_1-\rho_2)}{\bar{\rho}_2}\Bigg) \Bigg\}.
	\end{align}
	Under these conditions, replacing any occurrence of
	$n_1$ by $n_2$ at sufficiently large times strictly
	decreases the accumulated Dinkelbach cost.
	Hence, always choosing $n_2$ is optimal.
\end{proof}

\section{Proof of Lemma \ref{lem:OptimalStationaryAoI}}\label{app:OptimalStationaryAoI}
\begin{proof}
	When stationary policies are used, the \gls{aoi} process is a renewal reward process.
	Therefore, the average age can be written as
	\begin{align}
		\Delta_{\mathrm{avg}}(\xi) = \frac{\mathbb{E}[Q]}{\mathbb{E}[Y]},
	\end{align}
	where $Q$ and $Y$ denote the area under the age curve between two receptions and the duration of that interval, respectively.
	The area $Q$ can be written as
	\begin{align}
		Q = \int_{0}^{Y} (N_{prev}+t)dt = N_{prev}Y + \frac{Y^2}{2},
	\end{align}
	where $N_{prev}$ denotes the blocklength of the last successful update.
	Taking expectations yields
	\begin{align}
		\mathbb{E}[Q] = \mathbb{E}[N_{prev}]\cdot\mathbb{E}[Y] + \frac{\mathbb{E}[Y^2]}{2},
	\end{align}
	which follows from the independence of the chosen blocklength $N_{prev}$ and the transmission duration $Y$.
	The distribution of $N_{prev}$ can be calculated by conditioning the event of a successful transmission
	with blocklength $n_1$ ($n_2$, respectively) on the event of an overall successful transmission.
	This gives
	\begin{align}
		\mathbb{E}[N_{prev}] & = \frac{\xi\bar{\epsilon}_1n_1 + \bar{\xi}\bar{\epsilon}_2n_2}{\xi\bar{\epsilon}_1+\bar{\xi}\bar{\epsilon}_2}.
	\end{align}
	The expected value of $Y$ can be calculated via the recursion
	\begin{align}
		\mathbb{E}[Y] & = \xi(\bar{\epsilon}_1n_1 + \epsilon_1(n_1+\mathbb{E}[Y])) \notag \\
		              & + \bar{\xi}(\bar{\epsilon}_2n_2 + \epsilon_2(n_2+\mathbb{E}[Y])).
	\end{align}
	Solving for $\mathbb{E}[Y]$ gives
	\begin{align}
		\mathbb{E}[Y] = \frac{\xi n_1 + \bar{\xi}n_2}{\xi\bar{\epsilon}_1+\bar{\xi}\bar{\epsilon}_2}.
	\end{align}
	Similarly, for the expected value $\mathbb{E}[Y^2]$, we have the following recursion:
	\begin{align}
		\mathbb{E}[Y^2] & = \xi(\bar{\epsilon}_1n_1^2 + \epsilon_1\mathbb{E}[(n_1+Y)^2]) \notag \\
		                & + \bar{\xi}(\bar{\epsilon}_2n_2^2 + \epsilon_2\mathbb{E}[(n_2+Y)^2]).
	\end{align}
	Solving for $\mathbb{E}[Y^2]$ gives
	\begin{align}
		\mathbb{E}[Y^2]  = \frac{(\xi n_1^2 + \bar{\xi}n_2^2)}{\xi\bar{\epsilon}_1+\bar{\xi}\bar{\epsilon}_2} \notag
		+ \frac{2(\xi \epsilon_1 n_1 + \bar{\xi}\epsilon_2 n_2)(\xi n_1 + \bar{\xi}n_2)}{(\xi\bar{\epsilon}_1+\bar{\xi}\bar{\epsilon}_2)^2}.
	\end{align}
	Inserting everything and simplifying gives the result.
\end{proof}

\section{Proof of Lemma \ref{lem:OptimalBetaAoI}}\label{app:OptimalBetaAoI}

\begin{proof}
	Define the terms
	\begin{align}
		F(\xi) & \triangleq \xi n_1 + \bar{\xi}n_2,                          \\
		S(\xi) & \triangleq \xi\bar{\epsilon}_1 + \bar{\xi}\bar{\epsilon}_2, \\
		K(\xi) & \triangleq \xi n_1^2 + \bar{\xi}n_2^2.
	\end{align}
	With this, the average age expression from
	Lemma~\ref{lem:OptimalStationaryAoI} can be written as
	\begin{align}
		\Delta(\xi) = \frac{F(\xi)}{S(\xi)} + \frac{K(\xi)}{2F(\xi)}.
	\end{align}
	Differentiating with respect to $\xi$ gives
	\begin{align}
		\Delta'(\xi)
		 & =
		\frac{F'(\xi)S(\xi)-F(\xi)S'(\xi)}{S(\xi)^2}\notag \\
		 & +
		\frac{K'(\xi)F(\xi)-K(\xi)F'(\xi)}{2F(\xi)^2}.
	\end{align}
	Since
	\begin{align}
		F'(\xi) & =n_1-n_2,                           \\
		S'(\xi) & =\bar{\epsilon}_1-\bar{\epsilon}_2, \\
		K'(\xi) & =n_1^2-n_2^2,
	\end{align}
	we obtain after simplification
	\begin{align}
		F'(\xi)S(\xi)-F(\xi)S'(\xi)
		 & =
		n_1\bar{\epsilon}_2-n_2\bar{\epsilon}_1,
	\end{align}
	and
	\begin{align}
		K'(\xi)F(\xi)-K(\xi)F'(\xi)
		=
		n_1n_2(n_1-n_2).
	\end{align}
	Hence,
	\begin{align}
		\Delta'(\xi)
		=
		\frac{n_1\bar{\epsilon}_2-n_2\bar{\epsilon}_1}{S(\xi)^2}
		+
		\frac{n_1n_2(n_1-n_2)}{2F(\xi)^2}.
	\end{align}
	Any interior optimum $\xi \in (0,1)$ must satisfy the first-order condition $\Delta'(\xi)=0$.
	Multiplying by $2F(\xi)^2S(\xi)^2$ yields the quadratic equation
	\begin{align}
		 & 2\bigl(n_1\bar{\epsilon}_2-n_2\bar{\epsilon}_1\bigr)\bigl(\xi n_1+\bar{\xi}n_2\bigr)^2 \notag \\
		 & +
		n_1n_2(n_1-n_2)\bigl(\xi \bar{\epsilon}_1+\bar{\xi}\bar{\epsilon}_2\bigr)^2 = 0.
	\end{align}
	Therefore, every interior optimizer must be one of the roots $\tilde{\xi}_1,\tilde{\xi}_2$ of this quadratic equation.

	Since $\Delta(\xi)$ is continuous on the closed interval $[0,1]$,
	a global minimizer exists and must be attained either at a boundary point $\xi\in\{0,1\}$ or at an interior stationary point.
	Consequently, the optimal parameter $\xi_{\mathrm{avg}}^*$ is obtained by
	evaluating $\Delta(\xi)$ at $\xi=0$, $\xi=1$, and at
	the roots $\tilde{\xi}_1,\tilde{\xi}_2$
	that lie in $[0,1]$, and selecting the value of $\xi$ that gives the smallest average \gls{aoi}.
\end{proof}

\section{Proof of Theorem \ref{lem:ConstantSupoptimalityCertificate}}\label{app:ConstantSupoptimalityCertificate}
\begin{proof}
	Fix the parameter $\xi = \xi_{\mathrm{avg}}^*$ and let $\Delta_{\mathrm{avg}}(\xi_{\mathrm{avg}}^*)$ denote the average \gls{aoi} of the corresponding
	stationary randomized policy.
	By the renewal reward theorem \cite{gallagerStochasticProcessesTheory2013}, we have
	\begin{align}
		\lim_{T \to \infty} \Delta_{\mathrm{avg},T}(\bm{A}, \bm{Z}) = \Delta_{\mathrm{avg}}(\xi_{\mathrm{avg}}^*) \quad \text{a.s.},
	\end{align}
	where $\Delta_{\mathrm{avg},T}(\bm{A}, \bm{Z})$ denotes the average age over a sample path of
	actions $\bm{A}$ and transmission errors $\bm{Z}$ up to transmission index $T$.
	We now define the event
	\begin{align}
		G_{\xi_{\mathrm{avg}}^*} = \left \{ (\bm{a},\bm{z}): \lim_{T \to \infty} \Delta_{\mathrm{avg},T}(\bm{a}, \bm{z})
		= \Delta_{\mathrm{avg}}(\xi_{\mathrm{avg}}^*)\right \},
	\end{align}
	so that the renewal reward theorem is equivalent to
	\begin{align}\label{eq:RenewalRewardProb1}
		P(G_{\xi_{\mathrm{avg}}^*}) = 1.
	\end{align}
	Next, fix an action sequence $\bm{a} = (a_0, a_1, \ldots)$ and define the probability
	\begin{align}
		h(\bm{a}) = P_{\bm{Z}\vert \bm{a}} \left( \lim_{T \to \infty}  \Delta_{\mathrm{avg},T}(\bm{a}, \bm{Z})
		= \Delta_{\mathrm{avg}}(\xi_{\mathrm{avg}}^*) \right),
	\end{align}
	with the convention that the probability
	is zero if the limit does not exist.
	By applying the Fubini-Tonelli theorem to
	(\ref{eq:RenewalRewardProb1}), we get
	\begin{align}
		1 = P((\bm{A},\bm{Z}) \in G_{\xi_{\mathrm{avg}}^*}) = \int h(\bm{a})dP_{\bm{A}}(\bm{a}).
	\end{align}
	Because $0 \leq h(\bm{a}) \leq 1$, it follows that $h(\bm{a}) = 1$ on a set with probability one
	with respect to the measure $P_{\bm{A}}$ of the action sequence $\bm{A}$.
	Therefore, we can pick one deterministic action sequence $\bm{a}^*$ such that $h(\bm{a}^*) = 1$ and
	be guaranteed that
	\begin{align}
		P_{\bm{Z}\vert \bm{a}^*} \left(   \lim_{T \to \infty} \Delta_{\mathrm{avg},T}(\bm{a}^*, \bm{Z})
		= \Delta_{\mathrm{avg}}(\xi_{\mathrm{avg}}^*) \right) = 1.
	\end{align}
	Finally, by applying the dominated convergence theorem, we obtain
	\begin{align}
		 & \lim_{T \to \infty} \mathbb{E}_{\bm{Z}}\left[   \Delta_{\mathrm{avg},T}(\bm{a}^*, \bm{Z}) \right] \\
		 & = \mathbb{E}_{\bm{Z}}\left[ \lim_{T \to \infty} \Delta_{\mathrm{avg},T}(\bm{a}^*, \bm{Z}) \right]
		= \Delta_{\mathrm{avg}}(\xi_{\mathrm{avg}}^*).
	\end{align}
	To establish dominated convergence, we have to
	construct a single integrable random variable that
	dominates the finite-horizon average age
	uniformly over all sufficiently large $T$.
	After fixing the deterministic action sequence $\bm{a}^*$,
	let $p_t = \bar{\epsilon}(a_t^*)$.
	For the construction of the dominating random variable,
	we introduce the following probability space:
	\begin{align}\label{eq:CouplingProbabilitySpace}
		(\Omega,\mathcal{F},\mathbb{P})
		=
		\left(
		[0,1]^{\mathbb{N}_0},
		\mathcal{B}([0,1])^{\otimes\mathbb{N}_0},
		\operatorname{Leb}^{\otimes\mathbb{N}_0}
		\right),
	\end{align}
	where $\mathcal{B}([0,1])$ is the Borel $\sigma$-algebra on $[0,1]$ and $\operatorname{Leb}$ is the Lebesgue measure on $[0,1]$.
	Let $U_t(\omega)$ be the $t$-th coordinate map, so that the random variables
	$(U_t)_{t\geq 0}$ are independent and uniformly distributed on $[0,1]$.
	On this probability space, define the successful transmission
	indicators of the original process
	and of a coupled process, respectively, by
	\begin{align}
		S_t(\omega)
		 & = \mathbf{1}\{U_t(\omega)\leq p_t\},              \\
		\widetilde{S}_t(\omega)
		 & = \mathbf{1}\{U_t(\omega)\leq \bar{\epsilon}_1\}.
	\end{align}
	The process $(S_t)_{t\geq 0}$ realizes the
	process of successful transmissions associated with the
	conditional error distribution $P_{\bm{Z}\vert\bm{a}^*}$. We denote the
	corresponding error realization by $\bm{Z}(\omega)$.
	In contrast,
	$(\widetilde{S}_t)_{t\geq 0}$ is an i.i.d.\ Bernoulli process with success
	probability $\bar{\epsilon}_1$.
	Since $\bar{\epsilon}_1 \leq p_t$ for every $t$, the coupling satisfies
	\begin{align}
		\widetilde{S}_t(\omega)\leq S_t(\omega)
		\qquad
		\text{for every }t\geq 0\text{ and }\omega\in\Omega.
	\end{align}
	Thus, every successful transmission of the coupled
	process is also a successful transmission of the original
	process. In the following, the finite horizon is indexed by the integer $T$,
	where $T \in \{0,1,2,\ldots\}$, so that the corresponding time horizon
	is given by the sum of the blocklengths of the deterministic action sequence
	$\bm{a}^*$ up to index $T$, i.e., $\sum_{i=0}^{T} a_i^*$.
	Let $K(T)$ and $\widetilde{K}(T)$ denote the numbers of successful transmissions
	up to transmission index $T$ in the original and coupled processes, respectively. Let
	$a_{i-}^*$ be the blocklength of the update received at the beginning
	of the $i$-th inter-reception interval of the original process, and let $Y_i$ be the duration of its $i$-th
	inter-reception interval. Furthermore, let $\zeta_i$ be the number of transmission
	attempts in the $i$-th inter-success interval of the coupled process. The random
	variables $(\zeta_i)$ are i.i.d.\ geometric random variables with success
	probability $\bar{\epsilon}_1$.
	Let the term $R_T$ denote the residual age accumulated
	after the last successful transmission before
	the time horizon ends.
	Using the quantities introduced above, for every $T\geq 1$ and every
	$\omega\in\Omega$, we can upper bound the finite-horizon average age as follows:
	\begin{align}
		 & \Delta_{\mathrm{avg},T}(\bm{a}^*\!\!,\bm{Z}(\omega))                     \\
		 & = \frac{1}{\sum_{i=0}^{T}a_i^*}\left( \sum_{i=0}^{K(T)-1}\!\!
		\!\left( a_{i-}^*Y_i +
		\frac{Y_i^2}{2} \right)\! + \! R_T\right)                                   \\
		 & \overset{(a)}{\leq} n_2 + \frac{1}{(T+1) n_1}
		\sum_{i=0}^{\widetilde{K}(T)}\frac{n_2^2 \zeta_i^2}{2}                      \\
		 & \overset{(b)}{\leq} n_2 + \frac{n_2^2}{n_1}\frac{1}{T+1}\sum_{i=0}^{T+1}
		\frac{\zeta_i^2}{2}                                                         \\
		 & {\leq} n_2 + \frac{n_2^2}{2n_1}
		\sup_{T \geq 1} \frac{1}{T+1} \sum_{i=0}^{T+1} \zeta_i^2                    \\
		 & \overset{(c)}{\leq} \underbrace{n_2 + \frac{n_2^2}{2n_1}
			\left(\zeta_0^2+\sup_{N \geq 1}
			\frac{1}{N}\sum_{i=1}^{N}\zeta_i^2\right)}_{\triangleq D(\omega)}.\label{eq:DOmegaDefinition}
	\end{align}

	We now justify the steps in the chain of inequalities:
	For (a), first note that $a_{i-}^*\leq n_2$
	and $\sum_{i=0}^{T}a_i^*\geq (T+1)n_1$. Hence,
	the term $a_{i-}^*Y_i$ is bounded
	by $n_2Y_i$, and dividing by the total time horizon
	gives a contribution of at most
	$n_2$ to the average age.
	To upper-bound the quadratic terms,
	note that
	each successful reception of the coupled process is also
	a successful reception of the original process.
	The term $R_T$ vanishes, because
	the summation is taken up to $(\widetilde{K}(T))$,
	which includes the last inter-reception
	interval that started before the horizon at transmission index $T$ ended.
	Step (b) follows from $\widetilde{K}(T)\leq T+1$
	and $\zeta_i > 0$.
	Finally, step (c) follows from
	\begin{align}
		\frac{1}{M+1}\sum_{i=0}^{M+1}\zeta_i^2
		\leq \zeta_0^2+\frac{1}{M+1}
		\sum_{i=1}^{M+1}\zeta_i^2
	\end{align}
	which holds for every $M\geq 1$. This gives
	the dominating random variable $D(\omega)$ defined
	in \eqref{eq:DOmegaDefinition}.
	To show that $D(\omega)$ is integrable, it remains
	to show that $\sup_{N \geq 1} \frac{1}{N} \sum_{i=1}^{N} \zeta_i^2$ is integrable, because
	the term $\zeta_0^2$ is a squared geometric random variable and therefore integrable.
	In the following, we denote by $W_i = \zeta_i^2$ the i.i.d.\ squared geometric random variables
	with success probability $\bar{\epsilon}_1$.
	First, note that all the moments of a geometric random variable are finite.
	Let $\widehat{W}_i = W_i - \mathbb{E}[W_i]$ and $x > 2\mathbb{E}[W_i]$.
	Then, by applying the union bound, we get
	\begin{align}
		 & P\left( \sup_{N \geq 1} \frac{1}{N} \sum_{i=1}^{N} \zeta_i^2 > x \right)
		= P \left( \exists m: \frac{1}{m} \sum_{i=1}^{m} \zeta_i^2 > x \right)                                        \\
		 & \leq \sum_{m=1}^{\infty} P\left(  \sum_{i=1}^{m} \widehat{W}_i > m\left(x - \mathbb{E}[W_i]\right) \right) \\
		 & \leq \sum_{m=1}^{\infty} P\left( \sum_{i=1}^{m} \widehat{W}_i > mx/2 \right).
	\end{align}
	Now, we apply Markov's inequality and use
	$P(Y > x) \leq P(\vert Y \vert^4 > x^4)$ for $x \geq 0$ to get
	\begin{align}
		\sum_{m=1}^{\infty}\!\! P \! \left( \sum_{i=1}^{m}\! \widehat{W}_i > mx/2 \right)
		\!\leq \frac{2^4}{x^4}\! \sum_{m=1}^{\infty} \frac{1}{m^4}
		\mathbb{E}\!\left[ \left\vert \sum_{i=1}^{m} \widehat{W}_i \right\vert^4 \right]\!. \label{eq:MarkovBoundPY}
	\end{align}
	The last expectation can be upper-bounded as follows:
	\begin{align}
		 & \mathbb{E}\left[ \left\vert \sum_{i=1}^{m} \widehat{W}_i \right\vert^4 \right]
		= \sum_{i=1}^{m}\mathbb{E}[\widehat{W}_i^4] + 6 \sum_{1 \leq i < j \leq m} \mathbb{E}[\widehat{W}_i^2]\mathbb{E}[\widehat{W}_j^2]\notag \\
		 & =  m \mathbb{E}[\widehat{W}_1^4] + 3m(m-1)\mathbb{E}[\widehat{W}_1^2]^2 \leq C m^2,
	\end{align}
	with $0 \leq C < \infty$ being a constant that
	depends on the distribution of $W_i$.
	Here we used that the $\widehat{W}_i$ are i.i.d. and centered.
	Thus, by independence, every mixed term in
	which at least one index appears only once
	has expectation zero.
	Consequently, only terms of the
	form $\mathbb{E}[\widehat{W}_i^4]$ and
	$\mathbb{E}[\widehat{W}_i^2]\mathbb{E}[\widehat{W}_j^2]$, $i\neq j$,
	remain. The factor $6$ accounts
	for the number of permutations of the
	multiset $\{i,i,j,j\}$.
	Inserting this into (\ref{eq:MarkovBoundPY}) and using
	$\sum_{m=1}^{\infty} \frac{1}{m^2} = \frac{\pi^2}{6} < \infty$,
	which is known as the solution to the Basel problem, we get
	\begin{align}
		P \left(\underbrace{\sup_{N \geq 1} \frac{1}{N} \sum_{i=1}^{N} \zeta_i^2}_{\triangleq V} > x\right)
		\leq \tilde{C}x^{-4},
	\end{align}
	for a positive constant $\tilde{C}$, and
	$x > 2\mathbb{E}[W_i]$. Therefore, we can integrate
	$V$ and we get
	\begin{align}
		\mathbb{E}[V] & = \int_0^\infty P(V > x) dx                                                         \\
		              & \leq 2\mathbb{E}[W_i] + \int_{2\mathbb{E}[W_i]}^\infty \tilde{C}x^{-4} dx < \infty.
	\end{align}
	Since $\mathbb{E}[W_i] = \mathbb{E}[\zeta_0^2]<\infty$, this shows that the dominating function
	$D(\omega)$ is integrable, and we can apply the dominated convergence theorem.
	Note that if $\xi_{\mathrm{avg}}^* \in (0,1)$, then there exists a stationary randomized policy
	which achieves a smaller average \gls{aoi} than the two constant schedules of
	always choosing $n_1$ or always choosing $n_2$. In the previous steps, we have shown
	that we can choose a deterministic action sequence $\bm{a}^*$ that achieves the same average \gls{aoi}
	as the stationary randomized policy.
	This proves the theorem.
\end{proof}

\section{Proof of Lemma \ref{lem:StationaryRandomizedExpAoI}}\label{app:StationaryRandomizedExpAoI}

\begin{proof}
	The proof follows exactly the same lines
	as the proof of Lemma~\ref{lem:OptimalStationaryAoI}.
	By the renewal reward theorem, the average exponential age
	can be written as
	\begin{align}
		\Delta_{\exp,\mathrm{avg}}(\xi)
		=
		\frac{\mathbb{E}[Q^{\exp}]}{\mathbb{E}[Y]},
	\end{align}
	with $Q^{\exp}$ given by
	\begin{align}
		Q^{\exp}
		 & =
		\int_0^Y e^{\alpha(N_{\mathrm{prev}}+t)}dt
		=
		\frac{e^{\alpha N_{\mathrm{prev}}}}{\alpha}\bigl(e^{\alpha Y}-1\bigr).
	\end{align}
	Since $N_{\mathrm{prev}}$ and $Y$ are independent under
	the stationary randomized policy, we have
	\begin{align}
		\mathbb{E}[Q^{\exp}]
		=
		\frac{1}{\alpha}
		\mathbb{E}\!\left[e^{\alpha N_{\mathrm{prev}}}\right]\cdot
		\left(\mathbb{E}\!\left[e^{\alpha Y}\right]-1\right).
	\end{align}
	Conditioning on the event of an overall successful transmission gives
	\begin{align}
		\mathbb{E}\!\left[e^{\alpha N_{\mathrm{prev}}}\right]
		=
		\frac{\xi \bar{\epsilon}_1e^{\alpha n_1}+\bar{\xi}\bar{\epsilon}_2e^{\alpha n_2}}
		{\xi \bar{\epsilon}_1 + \bar{\xi}\bar{\epsilon}_2}
		=
		\frac{A(\xi)}{\xi \bar{\epsilon}_1 + \bar{\xi}\bar{\epsilon}_2}.
	\end{align}

	Next, let $M \triangleq \mathbb{E}\!\left[e^{\alpha Y}\right]$.
	Conditioning on the first transmission in the cycle yields
	\begin{align}
		M
		 & =
		\xi\left(\bar{\epsilon}_1e^{\alpha n_1}+\epsilon_1 e^{\alpha n_1}M\right) \\
		 & +
		\bar{\xi}\left(\bar{\epsilon}_2e^{\alpha n_2}+\epsilon_2 e^{\alpha n_2}M\right).
	\end{align}
	This can be written as
	$M = A(\xi) + B(\xi)M$,
	and solving for $M$ gives
	\begin{align}
		\mathbb{E}\!\left[e^{\alpha Y}\right]
		=
		M
		=
		\frac{A(\xi)}{1-B(\xi)},
	\end{align}
	provided that $B(\xi)<1$.
	From the proof of Lemma~\ref{lem:OptimalStationaryAoI}
	we get the expression for $\mathbb{E}[Y]$.
	Combining everything gives
	\begin{align}
		\Delta_{\exp,\mathrm{avg}}(\xi)
		=
		\frac{\mathbb{E}[Q^{\exp}]}{\mathbb{E}[Y]}
		=
		\frac{A(\xi)}{\alpha n_{\mathrm{avg}}(\xi)}
		\left(
		\frac{A(\xi)}{1-B(\xi)}-1
		\right),
	\end{align}
	as stated in the lemma.
\end{proof}

\section{Proof of Lemma \ref{lem:OptimalBetaExpAoI}}\label{app:OptimalBetaExpAoI}
\begin{proof}

	Any interior optimizer $\xi_{\exp}^* \in (0,1)$ must satisfy the first-order optimality condition
	\begin{align}
		\frac{d\Delta_{\exp,\mathrm{avg}}}{d\xi}(\xi_{\exp}^*)=0.
	\end{align}
	Differentiating the expression for the average exponential
	\gls{aoi} from Lemma~\ref{lem:StationaryRandomizedExpAoI} gives
	the equation in the lemma. Simplifying and
	sorting the coefficients
	shows that this is a quadratic equation in $\xi$.
	Finally, since $\Delta_{\exp,\mathrm{avg}}(\xi)$ is
	continuous on the feasible set
	$\{\xi\in[0,1]: B(\xi)<1\}$,
	a global minimizer is attained either at
	a boundary point or at an interior stationary point.
	Therefore, the optimizer can be found by evaluating
	the average exponential age at $\xi=0$, $\xi=1$
	and at all feasible
	roots of \eqref{eq:stationaryConditionCompactExp},
	and choosing the one that gives the smallest
	average exponential \gls{aoi}.
\end{proof}

\section{Proof of Theorem \ref{lem:ConstantSupoptimalityCertificateExpAoI}}\label{app:ConstantSupoptimalityCertificateExpAoI}

\begin{proof}
	We use the same selection argument as in the
	proof of Theorem \ref{lem:ConstantSupoptimalityCertificate}.
	Assume we are in a parameter regime with
	$\epsilon_1 e^{\alpha n_1} < 1$ and
	$\epsilon_2 e^{\alpha n_2} < 1$.
	Thus, for $\xi=\xi_{\exp}^*$, there exists a deterministic
	action sequence $\bm{a}^*$ such that the finite-horizon
	average exponential age converges almost surely to the
	average exponential age under the corresponding
	stationary randomized policy with parameter $\xi_{\exp}^*$.
	It remains to justify the exchange of limit
	and expectation via the dominated
	convergence theorem.
	Note that the approach from the proof
	of Theorem \ref{lem:ConstantSupoptimalityCertificate}
	does not work here, because the combination of the higher error
	probability with the larger blocklength leads to
	diverging expressions. Therefore,
	we use a martingale-based approach.

	Let $K(T)$ denote the number of successful receptions up to
	transmission index $T$ under the fixed sequence $\bm{a}^*$.
	As in Theorem \ref{lem:ConstantSupoptimalityCertificate}, let
	$Y_i$ be the duration of the $i$-th inter-reception interval
	and let $a_{i-}^*$ be the blocklength of the update received at
	the beginning of this interval. The exponential cost in
	this interval is given by
	\begin{align}
		Q_i^{\exp}
		\triangleq
		\int_0^{Y_i} e^{\alpha(a_{i-}^*+t)}\,dt .
	\end{align}
	In the following, the finite horizon is indexed by the integer $T$,
	where $T \in \{0,1,2,\ldots\}$, so that the corresponding time horizon
	is given by the sum of the blocklengths of the deterministic action sequence
	$\bm{a}^*$ up to index $T$, i.e., $\sum_{i=0}^{T} a_i^*$.
	Now, for every
	$T\geq 1$ and every sample path,
	the finite-horizon average
	exponential age can be bounded as follows:
	\begin{align}
		\Delta_{\exp,\mathrm{avg},T}(\bm{a}^*,\bm{Z})
		 & \overset{(a)} \leq
		\frac{1}{\sum_{t=0}^{T}a_t^*}
		\sum_{i=0}^{K(T)} Q_i^{\exp}     \\
		 & \overset{(b)}\leq
		\frac{e^{\alpha n_2}}{\alpha n_1(T+1)}
		\sum_{i=0}^{K(T)} e^{\alpha Y_i} \\
		 & \overset{(c)} \leq
		\underbrace{
			\frac{e^{\alpha n_2}}{\alpha n_1} \left(
			\sup_{N \geq 1}
			\frac{1}{N}\sum_{i=0}^{N}
			\underbrace{e^{\alpha Y_i}}_{\zeta_i} \right)
		}_{=:D_{\exp}(\omega)} .
	\end{align}
	Inequality (a) is valid because we have included
	the exponential cost of the last inter-reception interval
	which started before the horizon at transmission index $T$ ended.
	For the inequalities (b) and (c), we have
	used $\sum_{t=0}^{T}a_t^* \geq (T+1)n_1$,
	$K(T)\leq T+1$ and the non-negativity of the
	terms $e^{\alpha Y_i}$.
	To prove dominated convergence, we have
	to show that $\sup_{N \geq 1} \frac{1}{N} \sum_{i=0}^{N} \zeta_i$
	is integrable.
	Similar to the proof of Theorem \ref{lem:ConstantSupoptimalityCertificate},
	we first
	center the random variables $\zeta_i$ by
	subtracting their mean and bound the centered random variables
	and the mean separately.
	For this, let $\mathcal{F}_i$ denote the filtration
	generated by all information available up to the
	end of the
	$i$-th successful transmission. Note that
	this information determines the action sequence
	during the $i+1$-th inter-reception interval.
	Define the random variables
	\begin{align}
		 & \mathbb{E}[ \zeta_i  |\mathcal{F}_{i-1}]
		\triangleq \mu_i,                           \\
		 & M_i \triangleq \zeta_i - \mu_i.
	\end{align}
	With this, the
	random variable for which we want to show
	integrability can be written as
	\begin{align}
		\sup_{N \geq 1} \frac{1}{N} \sum_{i=0}^{N} \zeta_i    =
		\sup_{N \geq 1} \left( \frac{1}{N} \sum_{i=0}^{N} M_i + \frac{1}{N} \sum_{i=0}^{N} \mu_i\right). \label{eq:SumX_iSplit}
	\end{align}

	Next, we show that the
	expected values of the $\zeta_i$ are finite.
	With $\rho \triangleq \max_{j=1,2} \epsilon_j e^{\alpha n_j}$, the random variable $\mu_i$
	can be bounded as
	\begin{align}
		\mu_i
		= \mathbb{E}[ \zeta_i  |\mathcal{F}_{i-1}]
		\leq e^{\alpha n_2}\sum_{k=0}^\infty \rho^k
		\triangleq C_{\mu} < \infty.
	\end{align}
	To bound the $M_i$ terms, we will use Markov's inequality with an exponent greater than one.
	For this, we choose the exponent $q \in (1,2]$ such that
	\begin{align}
		\rho_q \triangleq\max_{j=1,2} \epsilon_j e^{\alpha n_j q} < 1.
	\end{align}

	With this, the $q$-th moment of $\zeta_i$ can be bounded as
	\begin{align}
		\mathbb{E}[\vert \zeta_i \vert ^q|\mathcal{F}_{i-1}]
		 & \leq
		e^{\alpha q n_2}\sum_{k=0}^\infty \rho_q^k
		\triangleq C_{\rho_q} < \infty.
	\end{align}
	With this and the Minkowski inequality, the $q$-th moment of $M_i$ can be bounded
	as
	\begin{align}
		 & \mathbb{E}[ \vert M_i \vert ^q|\mathcal{F}_{i-1}]^{1/q}                                                                  \\
		 & \leq \mathbb{E}[\vert \zeta_i \vert^q|\mathcal{F}_{i-1}]^{1/q} + \mathbb{E}[\vert \mu_i \vert^q|\mathcal{F}_{i-1}]^{1/q} \\
		 & \leq 2 \mathbb{E}[\vert \zeta_i \vert^q|\mathcal{F}_{i-1}]^{1/q},
	\end{align}
	where the second inequality follows from
	Jensen's inequality applied to the second term.
	Raising to the power of $q$ yields the bound
	\begin{align}
		\mathbb{E}[ \vert M_i \vert ^q|\mathcal{F}_{i-1}]
		\leq 2^q \mathbb{E}[\vert \zeta_i \vert^q|\mathcal{F}_{i-1}] \leq C_M < \infty.
	\end{align}
	Now, let
	$S_N \triangleq \sum_{i=0}^{N} M_i$.
	The stochastic process $S_N$ is a martingale:
	\begin{align}
		\mathbb{E}[M_i|\mathcal{F}_{i-1}]
		= \mathbb{E}[\zeta_i-\mu_i|\mathcal{F}_{i-1}] = 0.
	\end{align}
	We now show that
	$\sup_{N\geq 1}\left|\frac{1}{N}\sum_{i=0}^{N}M_i\right|$
	is integrable. Later we will need a bound on the expected value $\mathbb{E}[|S_N|^q]$,
	which can be obtained by
	the von Bahr-Esseen inequality
	with exponent
	$q\in(1,2]$ \cite{10.1214/aoms/1177700291}:
	\begin{align}
		\mathbb{E}\left[|S_N|^q\right]
		\leq 2 \sum_{i=0}^{N}\mathbb{E}\left[|M_i|^q\right].
	\end{align}
	Using the uniform conditional moment bound derived
	above, we get
	\begin{align}\label{eq:Bound_S_N_q}
		\mathbb{E}\left[|S_N|^q\right]
		\leq 2 C_M (N+1) \leq C_1 N,
	\end{align}
	for $N \geq 1$.
	We next use a refined union bound.
	For $x>0$,
	\begin{align}
		 & \mathbb{P}\left( \sup_{N\geq 1}\frac{|S_N|}{N}>x \right) \\
		 & \leq \sum_{m=0}^{\infty} \mathbb{P}\left( \max_{1
			\leq N\leq 2^{m+1}} |S_N| > 2^m x \right).
	\end{align}
	By Markov's inequality and Doob's
	$L^q$ maximal inequality,
	\begin{align}
		 & \mathbb{P}\left( \max_{1\leq N\leq 2^{m+1}} |S_N|
		> 2^m x \right)                                            \\
		 & \leq \frac{ \mathbb{E}\left[ \max_{1\leq N\leq 2^{m+1}}
		|S_N|^q \right] }{ 2^{mq}x^q }                             \\
		 & \leq \frac{ \left(\frac{q}{q-1}\right)^q
			\mathbb{E}\left[|S_{2^{m+1}}|^q\right] }
		{ 2^{mq}x^q }                                              \\
		 & \leq C_2\,2^{-m(q-1)}x^{-q},
	\end{align}
	where $C_2<\infty$ is independent
	of $m$ and $x$. In the last inequality, we have used the bound (\ref{eq:Bound_S_N_q}).
	Since $q>1$,
	\begin{align}
		\sum_{m=0}^{\infty}2^{-m(q-1)}<\infty.
	\end{align} Thus, there exists a
	finite constant $C_3$ such that
	\begin{align}
		\mathbb{P}\left( \sup_{N\geq 1}\frac{|S_N|}{N}>x
		\right) \leq C_3 x^{-q}.
	\end{align} Consequently,
	\begin{align}
		 & \mathbb{E}\left[ \sup_{N\geq 1}\frac{S_N}{N} \right]  \leq  \mathbb{E}\left[ \sup_{N\geq 1}\frac{|S_N|}{N} \right] \\
		 & = \int_0^\infty \mathbb{P}\left( \sup_{N\geq 1}\frac{|S_N|}{N}>x \right)\,dx                                       \\
		 & \leq 1+ \int_1^\infty C_3 x^{-q}\,dx <\infty.
	\end{align} Therefore,
	\begin{align}
		 & \mathbb{E}\left[ \sup_{N\geq 1} \frac{1}{N}\sum_{i=0}^{N}\zeta_i \right]       \\
		 & = \mathbb{E}\left[ \sup_{N\geq 1} \frac{1}{N}\sum_{i=0}^{N}(M_i+\mu_i) \right] \\
		 & \leq \mathbb{E}\left[ \sup_{N\geq 1}\frac{S_N}{N} \right]
		+ 2C_\mu <\infty.
	\end{align}
	Thus, we have constructed a dominating function and
	can interchange the limit and expectation by the dominated convergence theorem, which proves the theorem.
\end{proof}

\printbibliography

\end{document}